\documentclass[%
  superscriptaddress,
  preprint,
  nofootinbib,
  amsmath,amssymb,
  aps,
  pra,
  longbibliography,
]{revtex4-2}

\usepackage{hyperref}
\usepackage{amsthm}
\usepackage{dcolumn}%
\usepackage{bm}%

\usepackage{braket} %

\usepackage{booktabs} %

\usepackage{float}
\makeatletter
\let\newfloat\newfloat@ltx
\makeatother

\usepackage{algorithmicx}
\usepackage{algorithm}
\usepackage{algpseudocodex}

\usepackage{quantikz}

\graphicspath{{figures/}} %

\DeclareMathOperator{\Tr}{Tr}

\newcommand{\meqref}[1]{Eq.~\eqref{#1}}

\newtheorem{theorem}{Theorem}
\newtheorem{definition}[theorem]{Definition}

\newtheorem{lemma}[theorem]{Lemma}
\newtheorem{proposition}[theorem]{Proposition}
\newtheorem{remark}[theorem]{Remark}

\newcommand{\bigO}[1]{\ensuremath{\mathsf{O}\left(#1\right)}}
\newcommand{\Hilbert}{\ensuremath{\mathbb{H}}}
\newcommand{\GE}{\ensuremath{E_G}} %
\newcommand{\GEm}{\ensuremath{E'_G}}
\newcommand{\trueGE}{\ensuremath{\hat{E}_G}} %
\newcommand{\GM}{GE} %

\newcommand{\AVG}{\ensuremath{\overline{E}_G}}
\newcommand{\STD}{\ensuremath{\mathrm{IQR}}_{\AVG}} %
\newcommand{\AVGm}{\AVG'}
\newcommand{\STDm}{\ensuremath{\mathrm{IQR}}_{\AVGm}} %
\newcommand{\TState}[2]{#1[\ensuremath{#2}]}
\newcommand{\GHZ}[1]{\TState{GHZ}{#1}}
\newcommand{\W}[1]{\TState{W}{#1}}
\newcommand{\Ring}[1]{\TState{Ring}{#1}}
\newcommand{\Random}[1]{\TState{Random}{#1}}
\newcommand{\tomo}[1]{\ensuremath{\mathcal{T}_{#1}}} %
\newcommand{\ketZ}{\ket{{\bf 0}}} %

\newcommand{\setEG}{\ensuremath{\mathcal{F}}}
\newcommand{\setEGk}{\ensuremath{\setEG_k}}
\newcommand{\med}[1]{\ensuremath{\text{m}(#1)}}
\newcommand{\medSetEGk}{\ensuremath{\med{\setEGk}}}
\newcommand{\setEGm}{\ensuremath{\mathcal{F}'}}
\newcommand{\setEGmk}{\ensuremath{\setEGm_k}}
\newcommand{\medSetEGmk}{\ensuremath{\med{\setEGmk}}}
\newcommand{\imi}{\imath}
\newcommand{\complMap}{\mathfrak{v}}

\begin{document}

\title{A quantum implementation of high-order power method for
estimating geometric entanglement of pure states}%

\author{Andrii Semenov}
\email{andrii.semenov@equal1.com}
\author{Niall Murphy}%
\email{niall.murphy@equal1.com}
\affiliation{Equal1 Labs, Dublin, Ireland}%

\author{Simone Patscheider}
\affiliation{Department of Mathematics, University of Trento, Trento, Italy}%
\author{Alessandra Bernardi}
\affiliation{Department of Mathematics, University of Trento, Trento, Italy}%
\author{Elena Blokhina}
\affiliation{Equal1 Labs, Dublin, Ireland}%
\affiliation{School of Electrical \& Electronic Engineering,
University College Dublin, Ireland}%

\date{\today}

\begin{abstract}
  Entanglement is one of the fundamental properties of a quantum
  state and is a crucial differentiator between classical and quantum
  computation.
  There are many ways to define entanglement and its measure,
  depending on the problem or application under consideration.
  Each of these measures may
  be computed or approximated by multiple methods.
  However, hardly any of these methods can be run on near-term quantum hardware.
  This work presents a quantum adaptation of the iterative
  high-order power method for
  estimating the geometric measure of entanglement of multi-qubit
  pure states using  rank-1 tensor approximation.
  This method is executable on early fault-tolerant (hybrid) quantum hardware and
  does not depend on quantum memory.
  We simulate this algorithm 
  and mitigate the effects of noise on the results of the computation
  using a theoretical model based on a known mitigation approach, which assumes a global depolarising noise channel.
\end{abstract}

\maketitle

\section{Introduction}
\label{sec:introduction}

For a quantum algorithm to have an advantage over a classical alternative,
entanglement~\cite{Horodecki2009} is a necessary ingredient.
Not only the presence, but the ``degree''
of entanglement in a quantum state is an important property for many
applications~\cite{Wootters1998a}, for example quantum information
technologies~\cite{Guehne2009}, quantum teleportation and quantum
communication~\cite{Horodecki2009} and quantum
cryptography~\cite{Broadbent2016}.
On the one hand, if there is not ``enough'' entanglement, a quantum
circuit can be efficiently simulated by classical devices~\cite{Vidal2003}.
On the other hand, some systems with ``maximally'' entangled states,
such as stabiliser codes~\cite{Gottesman1998a},
admit efficient classical simulations~\cite{Gottesman1998,Aaronson2004}.
In addition, highly entangled states are not useful as computational
resources in the measurement-based computing paradigm~\cite{Gross2009}.
In the setting of parameterised quantum circuits, entanglement contributes to
``barren plateaus'' in the cost function landscape that make training
a challenge~\cite{Patti2021,Marrero2021}.

Consequently, there has been an effort to define and
quantify~\cite{Wootters1998a}
entanglement, with different methods~\cite{Georgiev2022} being
preferred depending on the application.
Some definitions (for example, quantum mutual information or von
Neumann entropy~\cite{Horodecki2009})
are measures of bipartite entanglement only and are computationally
expensive for mixed states.
Others, such as concurrence~\cite{Wootters1998}, have no unique
definition for higher ($>2$) dimensional systems~\cite{Love2007,Plbnio2007}.
In fact, the very concept of $n$-partite (multi-qubit) entanglement
is still not clearly understood~\cite{Duer2000}.

The geometric measure of entanglement
($\GE$)~\cite{Shimony1995,Barnum2001,Wei2003}
is a popular multipartite entanglement metric with a clear geometric
interpretation,
which naturally extends to the $n$-partite case and to mixed
states~\cite{Wei2003}.
Geometric entanglement is, among other applications, useful when
defining entanglement witnesses~\cite{Wei2003} which are used
as a positive indicator of entanglement for a subset of quantum states.

The first algorithm to compute $\GE$ used it (under the name of
``the Groverian measure of entanglement'') to assess the probability of success
for an initial state in Grover's algorithm~\cite{Shimoni2005,Most2010}.
This method was rephrased in terms of eigenvalues and Singular Value
Decomposition (SVD) and extended to mixed states~\cite{Streltsov2011}.
A similar method was formulated in terms of Tucker decomposition~\cite{Teng2017}
using the Higher Order Orthogonal Iteration (HOOI)
algorithm~\cite{DeLathauwer2000}.
A parallel strand of work~\cite{Hayashi2009,Ni2014}, based on the
fundamental connection between
geometric entanglement and tensor theory, introduced a series of
algorithms to approximate \GE~\cite{Hu2016,Qi2018,Zhang2020a}.
These approaches can be considered as generalisations of the popular
power methods (e.g.\,the Lanczos method and Arnoldi iterations) used
to study the eigenvalues of matrices (for example,
see~\cite{ChatelinEigenMatrix}).

With accessible quantum hardware, measuring entanglement on a
physical device becomes a problem of interest.
The algorithms mentioned above are all intended to be executed on classical
machines. Na\"{i}vely, we would have to do full-state tomography to
reconstruct the quantum state in a classical computer to calculate
its entanglement.
As the number of qubits increases, the number of measurements
required for tomography increases
exponentially. %
Alternatively, if we know the quantum state in advance, we can
prepare the appropriate entanglement witness~\cite{Guehne2009}.
However, the latter method does not allow one to obtain the measure
of entanglement itself.
Recently, variational quantum circuits (VQC) have been suggested to
compute geometric
entanglement of pure states on a quantum
computer~\cite{MunozMoller2022, Consiglio2022}.
VQC algorithms, unfortunately, suffer from ``barren
plateaus''~\cite{Patti2021,Marrero2021} which hinders the scalability
of the approach.

In this paper, we present an iterative quantum algorithm for
computing the geometric entanglement of pure quantum states.
The algorithm is a quantum adaptation of the High Order Power
Method (HOPM)~\cite{DeLathauwer2000} to find solutions for Rank-1
Tensor Approximation (RTA).
We show how to implement crucial steps of HOPM in the quantum domain
and analyse their robustness w.r.t.\@ noise.
Our aim is to measure entanglement on near-term quantum devices more
time efficiently than full-state tomography and more space
efficiently (qubits versus classical memory) than executing HOPM on a
classical device.

Our paper is structured as follows. First, we formally define
geometric entanglement and its connection to RTA and recall the HOPM
algorithm for this problem.
Next, we present our quantum implementation of HOPM\@.
Then we show some simulation results exploring the robustness of the
algorithm to noise and discuss some simple ways to mitigate it.
Finally, we discuss some future directions and unresolved questions.

\section{Preliminaries: Geometric Entanglement in terms of rank-1
tensor approximation}
\label{sec:preliminaries}
\begin{definition}[Fully separable and entangled states]
  \label{def:seperable}
  Let $\mathbb{H}^n$ be an $n$-qubit Hilbert space.
  A pure $n$-partite state $\ket{\phi} \in \Hilbert^n$ is \emph{fully
  separable} if and only
  if it is a product state of 1-qubit states $\ket{{\bf v}_i} \in
  \Hilbert$~\cite{Horodecki2009}: %
  \begin{equation}
    \label{eqn:seperable}
    \ket{\phi} = \ket{{\bf v}_1} \otimes \cdots \otimes \ket{{\bf v}_n}.
  \end{equation}
  We say an $n$-partite pure state is \emph{entangled} if it is not
  fully separable.
  Let $S_n \subseteq \Hilbert^n$ denote the set of fully separable states.
\end{definition}
\vspace{1em}
\begin{definition}[Geometric measure of
  entanglement~\cite{Shimony1995,Barnum2001,Wei2003}]
  \label{def:geometric_entanglement}
  Let the \emph{geometric measure of entanglement}
  of a pure state $\ket{\psi}$ be
  \begin{equation}\label{eqn:gm_def}
    \trueGE(\ket{\psi})=1 - \max_{\ket{\phi} \in S_n}|\langle \phi |
    \psi \rangle|^2,
  \end{equation}
  where $\hat{\lambda} = \max\limits_{\ket{\phi} \in S_n}|\langle
  \phi | \psi \rangle|$ is called the \emph{entanglement eigenvalue}.
\end{definition}

One approach to compute the value of $\hat{\lambda}$ is to solve the
problem of finding the ``closest'' fully separable state $\ket{\phi}$
to the state $\ket{\psi}$.
The closest separable state is formally stated as a minimization
problem over $\ket{\phi} \in S_n$ with the objective function being
the distance between $\ket{\psi}$ and $\ket{\phi}$:
\begin{equation}
  \label{eqn:gm_obj_func}
  \min_{\ket{\phi} \in S_n} d(\ket{\psi}, \ket{\phi})^2,
\end{equation}
where $d(a,b)=\|a-b\|_F$ is the reference distance with $a,b\in
\Hilbert^n$ and $\|\cdot\|_F$ being the Frobenius norm (we will omit
the subscript $F$ in what follows). We note that the minimiser of
\meqref{eqn:gm_obj_func} exists (but may not be unique) since the set
of fully separable states $S_n$ is the classical Segre
variety~\cite{Bernardi2018, cirici2020characterization,QPG}, which is
closed in Zariski and Euclidean topology when defined over the
complex numbers. Minimizing the distance from a variety is a problem
well studied from different perspectives, for example in the context
of algebraic geometry the concept of Euclidean Distance Degree has
been introduced~\cite{Draisma201699}.
In the present work we deal with the
minimization of Eq.~\eqref{eqn:gm_obj_func} in terms of a system of
nonlinear equations with $\lambda$ being a Lagrange multiplier
(see~\cite{Ni2014,Wei2003,Zhang2020a}):
\begin{equation}\label{eqn:nonlinear-system}
  \begin{aligned}
    & T_\psi \times_{{i}} \left({\bf v}_1^\ast,\ldots, {{\bf
      v}_{i-1}^\ast}, {\bf v}_{i+1}^\ast, \ldots, {\bf v}_n^\ast
    \right) = \lambda {\bf v}_i, \\
    & T_\psi^* \times_{{i}} \left({{\bf v}_1},\ldots, {{\bf
    v}_{i-1}}, {{\bf v}_{i+1}}, \ldots, {{\bf v}_n} \right)=\lambda
    {\bf v}_i^\ast,          \\
    & \| {\bf v}_i \|^2 = \sum_{b_i} |v_{i,b_i}|^2 = 1, \quad  i \in
    \{1,\ldots,n\}, \quad \lambda \in \mathbb{R},
  \end{aligned}
\end{equation}
where $T_\psi \in \mathbb{C}_2^{\otimes n}$ ($\mathbb{C}_2 =
  \mathbb{C}\times \mathbb{C}$ is a two-dimensional complex vector
space) is a tensor representation of $\ket{\psi}$ in the basis
$\ket{b_1,\ldots,b_n} = \otimes_i \ket{b_i}$ with components
$\psi_{b_1,\ldots,b_n} \in \mathbb{C}$  defined as:
\begin{equation*}\label{eqn:psi_phi-tensor-connection}
  \ket{\psi} = \sum_{b_1,\ldots,b_n} \psi_{b_1,\ldots,b_n} \ket{b_1,
  \ldots, b_n},
\end{equation*}
and ${\bf v}_i = (v_{i,0}, v_{i,1})^{\rm T} \in \mathbb{C}_2$ is a
vector of components of the one-qubit states $\ket{{\bf v}_i}$,
defined in \meqref{eqn:seperable}, in the same basis;
the asterisk means complex conjugate; 
the norm $\|\cdot\|$ is an $l^2$-norm;
the symbol $\times_{ i}$ is
$n$-mode vector product over all modes except the $i$-th one, which
is a contraction of a tensor on the left with the tuple of vectors on
the right, skipping the $i$-th index (mode) of the tensor (when the
subscript is not given none of the modes are skipped).
In tensor analysis, the problem of finding a solution to a system of
equations such as \meqref{eqn:nonlinear-system} is usually referred
to as the $U$-eigenpair problem for the tensor $T_\psi$.
One way to solve this problem is by reducing it to RTA~\cite{Ni2014,Zhang2020a}.

The problem of RTA is $\mathsf{NP}$-hard~\cite{Hillar2013} and so
exact or global techniques such
as homotopy continuation methods~\cite{Chen2016} may take exponential
time. %
Alternatively, approximate techniques may run quickly, but may return
a local minima.
Some of the well-known algorithms that can be used for approximating RTA are
Higher-Order Singular Value Decomposition (HOSVD)~\cite{DeLathauwer2000a},
HOPM and HOOI~\cite{DeLathauwer2000}.
In the particular case of tensors representing the pure quantum state
of a finite system of qubits, both HOSVD and HOOI simplify to the
HOPM algorithm, which is based upon the Alternating Least Squares
(ALS) approach~\cite{DeLathauwer2000,Kroonenberg1980} for solving
non-linear systems of equations.

The application of HOPM to the first sub-system of equations in
\meqref{eqn:nonlinear-system}
is presented in pseudocode in Algorithm~\ref{alg:hopm}.
(Note, that this is a known
technique~\cite{DeLathauwer2000,Streltsov2011,Ni2014,Qi2018,Zhang2020a}).

\begin{algorithm}
  \caption{HOPM algorithm for estimating RTA of $T_\psi$ with
  \label{alg:hopm}
  absolute accuracy $\epsilon$~\cite{DeLathauwer2000}.}
  \begin{algorithmic}[1]
    \Procedure{HOPM}{$T_\psi$, $\epsilon$}
    \State{Let
      $(\mathbf{v}_1^{(0)},\ldots,\mathbf{v}_n^{(0)})$,
      $\lambda^{(0)}$ be a tuple of some $n$ one-qubit state vectors
    and the corresponding entanglement eigenvalue.}
    \State{Initialise $k$ to $1$.}
    \While{\label{line:hopm_conv_cond}$\mid
    \lambda^{(k)}-\lambda^{(k-1)}\mid\ > \epsilon$}
    \For{$i \in {1,\ldots,n}$} \\
    \State{${\bf u}^{(k)}_{i} = T_\psi \times_{{i}} \left({{\bf
        v}^{(k)^\ast}_1},\ldots, {{\bf v}^{(k)^\ast}_{i-1}}, {\bf
      v}^{(k-1)^\ast}_{i+1}, \ldots, {\bf v}^{(k-1)^\ast}_n \right)
    $\label{line:hopm_v_update}} \\
    \State{${\bf v}^{(k)}_{i}  = {\bf u}^{(k)}_{i} / \Vert
    {\mathbf{u}}_i^{(k)} \Vert $\label{line:hopm_v_norm}} \\
    \EndFor
    \State{$\lambda^{(k)} = \left| T_\psi \times \left({{\bf
      v}^{(k)^\ast}_1}, \ldots, {\bf v}^{(k)^\ast}_n \right) \right|
    $\label{line:hopm_nmode_prod}} \\
    \State{increment $k$}
    \EndWhile
    \Return{$\lambda^{(k)}$ and
    $(\mathbf{v}_1^{(k)},\ldots,\mathbf{v}_n^{(k)})$}
    \EndProcedure
  \end{algorithmic}
\end{algorithm}

\begin{theorem}%
\label{thrm:hopm-conv}
    Let $\lambda^{(k)}$ and $({{\bf v}^{(k)}_1},\ldots, {\bf v}^{(k)}_n )$ be the output of Algorithm~\ref{alg:hopm} at the $k$th iteration for the given tensor representation $T_\psi \in \mathbb{C}_2^{\otimes n}$ of a state $\ket{\psi}$. Then the number of iterations $k$ needed to get to the desired accuracy $\epsilon = \hat{\lambda}-\lambda^{(k)}$ scales with $n$ as
    \begin{equation}\label{eqn:hopm-compl-conv-rate-psi}
        \bigO{{n^2}/{\epsilon}}.
    \end{equation}
\end{theorem}
\begin{proof}
    From Theorem~\ref{thrm:hopm-sublin-conv} (see Appendix~\ref{app:hopm-convergence-proof}) on the sublinear convergence of HOPM  we have
    \begin{equation*}
        \epsilon \leq B{\left( \frac{p-2}{n^2 p}k \right)^{-\frac{p}{p-2}} },
    \end{equation*}
    where $p=n(3n-3)^{4n}$ for the tensor $T_\psi$. For large $n \gg 1$ it reduces to
    \begin{equation}\label{eqn:hopm-compl-conv-rate-qubits}
        k \leq B{\frac{n^2}{\epsilon} }\quad \in \bigO{n^2/\epsilon}.
    \end{equation}
\end{proof}

\begin{remark}
  []\label{prop:global-phases}
  The global phase of $\ket{{\bf v}_i}$ and the chosen
  computational basis do not affect the result for $\lambda$ and
  convergence of Algorithm~\ref{alg:hopm}.
\end{remark}
Indeed from line~\ref{line:hopm_nmode_prod} of Algorithm~\ref{alg:hopm} it
is evident that global phase prefactors $e^{i \alpha_j}$ ($\alpha_j \in
\mathbb{R}$) of ${{\bf v}^{(k)}_{j}}$ do not affect the result
$\lambda$ of the
algorithm on any iteration $k$. From the same line we see that the
change of basis does not change $\lambda$, since the inner product
of $T_\psi$ with all the vectors ${\bf v}_i$ on $k$-th iteration is
basis invariant. This also means that the convergence of the
algorithm, which is determined by $\lambda$ (see
Theorem~\ref{thrm:hopm-conv}), is also not affected.

\section{Quantum implementation of HOPM for estimating the
entanglement eigenvalue}
\label{ssec:our_implementation}

In this section, we present an iterative quantum algorithm
(Algorithm~\ref{alg:quals}), which is a HOPM approach for
approximating RTA with its crucial steps
(lines~\ref{line:hopm_v_update}-\ref{line:hopm_nmode_prod} in
Algorithm~\ref{alg:hopm}) performed on a quantum device. We will
refer to this algorithm as QHOPM\@.
The input for the algorithm is an $n$-qubit unitary operator
$U_\psi \in \mathrm{U}(2^n)$ that prepares the target state
$\ket{\psi} = U_\psi \ket{\bf 0} \in \Hilbert^n$, where $\ket{\bf
0} = \ket{0}^{\otimes n}$ is the initial state of the system.
The separable state $\ket{\phi} \in S_n$ is encoded as a tensor
product of one-qubit $x$ and $z$ rotations acting on $\ket{\bf 0}$:
\begin{equation}\label{eq:phi-encode}
  \ket{\phi} = \bigotimes_{i=1}^n \ket{{\bf v}_i}
  = \bigotimes_{i=1}^n R_z(\varphi_i) R_x (\vartheta_i) \ket{\bf 0},
\end{equation}
where $\ket{\mathbf{v}_i}$ are one-qubit states,
and $\vartheta_i \in [0, \pi), \varphi_i \in [0, 2\pi)$ are the
angles used to encode $\ket{\phi}$. This choice of encoding of
$\ket{\mathbf{v}_i}$ differs from any other encoding only by global
phase and therefore is valid due to Remark~\ref{prop:global-phases}.
We choose an initial separable state $\ket{\phi^{(0)}}$ by randomly
choosing the angles $(\vartheta_i^{(0)}, \varphi_i^{(0)})$ as a
starting point for the approximation.

The first key steps of Algorithm~\ref{alg:hopm} are
lines~\ref{line:hopm_v_update} and~\ref{line:hopm_v_norm}
(consisting of the $n$-mode product and normalization)
which
yield an updated version of $\ket{\mathbf{v}_i}$.

\begin{figure} %
  \centering
    \includegraphics{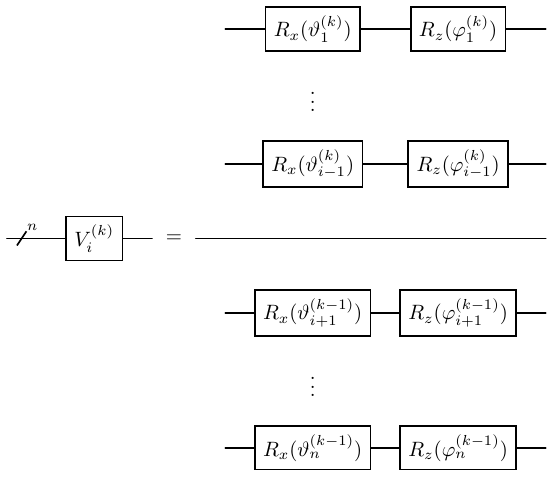}
  \caption{\textbf{
      Circuit representation of a separable state used for the $i$th
    mode update.}
    The sub-circuit implements  $V_i^{(k)}$ (see
    \meqref{eqn:tensor_contract_circuit}),
    the vector part of the $n$-mode vector product for
    line~\ref{line:hopm_v_update} of Algorithm~\ref{alg:hopm}.
  }
  \label{fig:v_n}
\end{figure}
\begin{proposition}\label{prop:alg-equiv}
  Lines~\ref{line:hopm_v_update} and~\ref{line:hopm_v_norm} of
  Algorithm~\ref{alg:hopm} for the state $\ket{\psi} = U_\psi
  \ket{\bf 0}$ at any iteration  $k>0$ are equivalent, up to a
  global phase, to the following update of $\ket{\mathbf{v}_i^{(k)}}$:
  \begin{equation}\label{eqn:u_i-normalized}
    \ket{{\bf v}_i^{(k)}} = \sum_b v_{i,b}^{(k)} \ket{b}
    = \frac{\sum_b u_{i,b}^{(k)} \ket{b}}{\sqrt{\sum_b | u_{i,b}^{(k)} |^2}},
  \end{equation}
  where
  \begin{equation}\label{eqn:u_i_j-th_component}
    u_{i,b}^{(k)} =  \bra{b_{[i]}} {V_i^{(k)}}^\dagger U_\psi \ket{\bf 0},
    \quad
    \ket{b_{[i]}} = \left( \ket{0}^{\otimes (n-i)} \otimes \ket{b}
    \otimes \ket{0}^{\otimes (i-1)}  \right)
  \end{equation}
  is the $b$-th component of the $n$-mode vector product in
  line~\ref{line:hopm_v_update} having
  \begin{equation}
    \label{eqn:tensor_contract_circuit}
    \begin{split}
      V_i^{(k)} = & \left[ \bigotimes_{l=i+1}^{n}
      R_z(\varphi_l^{(k)}) R_x(\vartheta_l^{(k)}) \right] \\%
      & \otimes I \otimes
      \\%
      & \left[ \bigotimes_{l=1}^{i-1} R_z(\varphi_l^{(k-1)})
      R_x(\vartheta_l^{(k-1)}) \right],
    \end{split}
  \end{equation}
  (see Fig.~\ref{fig:v_n}), and $\ket{b}$ is a one-qubit computational basis.
\end{proposition}
\begin{proof}
  We start with line~\ref{line:hopm_v_update}
  Algorithm~\ref{alg:hopm} where we calculate the updated ${\bf
  u}_i^{(k)}$, the $b$-th component of which is
  \begin{align*}
    u_{i,b}^{(k)}  = T_\psi \times \left(
      {\textbf{v}_1^{(k)}}^\ast, \ldots
      ,{\textbf{v}_{i-1}^{(k)}}^\ast, {\bf e}^\ast_b,
      {\textbf{v}_{i+1}^{(k-1)}}^\ast, \ldots,
    {\textbf{v}_n^{(k-1)}}^\ast\right),
  \end{align*}
  where ${\bf e}_b$ is a one-qubit basis vector for the $b$-th component.
  Let us substitute to the right-hand side the following
  expressions for the components of $T_\psi$ and ${\bf v}_{m \neq
  i}^{(t)}$ in terms of $U_\psi$ and $R_zR_x$:
  \[
  \psi_{l_1,\ldots,l_n} = \bra{l_1 \ldots l_n} U_\psi \ket{\bf 0};\quad
  v_{m,l}^{(t)} = \bra{l} R_z (\varphi_m^{(t)}) R_x
  (\vartheta_m^{(t)}) \ket{0},
  \]
  and taking into account that the $l$-th component of  ${\bf e}_b$
  in bra-ket notation is $\bra{l} b \rangle$,
  we get for the $b$-th component of $\textbf{u}_i^{(k)}$
  \begin{equation*}%
    \begin{split}
      u_{i,b}^{(k)} = & \sum_{l_1,\ldots,l_n}  \bra{0}  R_x^\dagger
      (\vartheta_{1}^{(k)}) R_z^\dagger (\varphi_{1}^{(k)})
      \ket{l_1} \ldots\langle b \ket{l_i} \ldots \\
      & \bra{0} R_x^\dagger (\vartheta_{n}^{(k-1)}) R_z^\dagger
      (\varphi_{n}^{(k-1)}) \ket{l_n} \bra{l_1 \ldots l_n} U_\psi
      \ket{0}                   \\
      =                 & \bra{b_{[i]}} {V_i^{(k)}}^\dagger U_\psi
      \ket{\bf 0},
    \end{split}
  \end{equation*}
  where $V_i^{(k)}$ is defined by \meqref{eqn:tensor_contract_circuit}.
  This expression corresponds to Eq.~(\ref{eqn:u_i_j-th_component}).

  To prove that Eq.~(\ref{eqn:u_i_j-th_component}) yields
  line~\ref{line:hopm_v_update} of Algorithm~\ref{alg:hopm} we
  perform the above steps in reverse, taking into account
  Remark~\ref{prop:global-phases}.

  Note, that for the components defined by
  \eqref{eqn:u_i_j-th_component} to correspond to a proper quantum
  state, they need to be normalized:
  \begin{equation*}%
    \ket{{\bf v}_i^{(k)}} = \frac{\sum_b u_{i,b}^{(k)}
    \ket{b}}{\sqrt{\sum_b | u_{i,b}^{(k)} |^2}},
  \end{equation*}
  which corresponds to the line~\ref{line:hopm_v_norm} of
  Algorithm~\ref{alg:hopm}.
\end{proof}

Proposition~\ref{prop:alg-equiv} allows us to understand which type
of operations and measurements we should perform to implement
lines~ \ref{line:hopm_v_update} and~\ref{line:hopm_v_norm} of the
algorithm, in particular, we need to be able to obtain the updated
pair of angles $(\vartheta_i^{(k)}, \varphi_i^{(k)})$ for further
encoding the updated $\ket{{\bf v}_i^{(k)}} =
R_z(\varphi_i^{(k)})R_x(\vartheta_i^{(k)})\ket{0}$.

\subsubsection{Recovering one-qubit states with
tomography}\label{sec:one-qubit_tomography}
To obtain the angles $(\vartheta_i^{(k)}, \varphi_i^{(k)})$ for the
one-qubit state $\ket{{\bf v}_i^{(k)}}$
of the $i$-th qubit at an arbitrary iteration  $k$ of QHOPM, we use
the following one-qubit tomography procedure.
Let $\tomo{i}\colon  \mathrm{U}(2^n) \to [0, \pi) \times [0, 2\pi)$,
which for some $W$ returns the angles  $(\vartheta_i, \varphi_i)$,
such that for $\ket{w} = W\ketZ$, $\ket{q_i} =
R_z(\varphi_i)R_x(\vartheta_i) \ket{0}$ and any one-qubit basis
state $\ket{s}$:
\begin{equation}
  \label{eqn:coeff_one_qubit_tomography}
  \braket{s | q_i} = \frac{\braket{s_{[i]} | w}}%
  {\sqrt{%
    |\braket{s_{[i]} | w}|^2 + |\braket{ s^\perp_{[i]} | w}|^2}%
  },
\end{equation}
and $\ket{s^\perp_{[i]}}$ is an orthogonal state to $\ket{s_{[i]}}$
defined in~\meqref{eqn:u_i_j-th_component}.

In our implementation, $\tomo{i}$ solves (using a classical device)
the following system of equations
\begin{equation}\label{eq:tomography-system}
  \begin{aligned}
    \braket {Z}_i  &= 2P_i(0, W) - 1   = \cos \vartheta_i;                \\
    \braket {X}_i  &= 2P_i(+, W) - 1 = \sin \vartheta_i \sin \varphi_i; \\
    \braket {Y}_i  &= 2P_i(\imi, W) - 1 = -\sin \vartheta_i \cos \varphi_i, \\
  \end{aligned}
\end{equation}
where $\braket {A}_i = \bra{q_i} A \ket{q_i}$;
for the quantum system prepared in the state $\ket{w}$,
\begin{equation}\label{eqn:tomography-prob-s}
  P_i(b, W)
  =  | \braket{b | q_i} |^2
  = \frac{|\braket {b_{[i]} | w}|^2}{|\braket{ b_{[i]} | w}|^2 +
  |\braket{ b^\perp_{[i]} | w}|^2}
\end{equation}
is a probability of the $i$-th qubit being in the state $\ket{b}$,
whilst other qubits are in the state $\ket{0}$.
The probabilities $P_i(b, W)$ are obtained by querying the quantum system.

One way to obtain $P_i(b, W)$ is by direct measurement of $|\langle
b_{[i]}\ket{w}|^2$,
however the number of shots for a given accuracy in this approach
grows exponentially with the number of qubits.
In our implementation, we use a Hadamard test-based~\cite{Cleve1998} procedure
where the number of shots depnds only on the accuracy and not the
number of qubits.

First, we classically reconstruct the coefficients $C_s =
\braket{s|q_i}$  (defined by
\meqref{eqn:coeff_one_qubit_tomography}) for the state
\[
\ket{q_i} = C_b \ket{b} + C_{b^\perp} \ket{b^\perp}
\]
in some basis $B_b = \{ \ket{b}, \ket{b^\perp} \}$.
For this we measure the real and imaginary parts of $ \braket{s_{[i]}| w}$
for each basis state $\ket{s} \in B_b$,
using the Hadamard test procedure as follows.
\begin{enumerate}
  \item Introduce an ancilla qubit, initialised in the state $\ket{+}$.
    The other qubits (referred to as data qubits in what follows)
    are initialised in the ground state $\ket{\bf 0}$.
  \item Perform a unitary operation $W$ on the data qubits,
    controlled by the ancilla qubit.
  \item Perform a unitary operation $U_{s}^\dagger$, that
    transforms $\bra{0}$ to $\bra{s}$ state, on the $i$-th data
    qubit, controlled by the ancilla qubit.
  \item Measure $x_a = \braket{X_a}$ and $y_a = \braket{Y_a}$ on
    the ancilla qubit  ($A_a = A \otimes I^n$).
  \item Calculate $\braket{s_{[i]}| w} = x_a + \imi y_a$
    ($\imath$ is an imaginary unit).
\end{enumerate}

In the case of QHOPM,
$W = {V_i}^{(k)\dagger}U_\psi$ and $\ket{q_i} = \ket{{\bf
v}_i^{(k)}}$ as defined in \meqref{eqn:u_i-normalized}.

\subsubsection{Measuring Entanglement}
\label{ssec:the_final_algorithm}
At the end of each iteration $k$, we need to measure $\lambda^{(k)}$.
\begin{remark}\label{prop:lambda-circ-equiv}
  Line~\ref{line:hopm_nmode_prod} of Algorithm~\ref{alg:hopm} for a
  target state $\ket{\psi} = U_\psi \ket{\bf 0}$ at any $k>0$
  iteration is equivalent to:
  \begin{equation}\label{eqn:lambda-quantum}
    \lambda^{(k)} = | \bra{\bf 0} {V^{(k)}}^\dagger U_\psi \ket{\bf 0} |,
  \end{equation}
  where
  \begin{equation}
    \label{eqn:inner_prodcut_circuit}
    V^{(k)} = \bigotimes_{i=1}^{n} R_z(\varphi_i^{(k)})
    R_x(\vartheta_i^{(k)}).
  \end{equation}
\end{remark}
Indeed, using the same procedure as in
Proposition~\ref{prop:alg-equiv} we find:
\begin{equation*}%
  \lambda^{(k)} = |\bra{\bf 0} {V^{(k)}}^\dagger U_\psi \ket{\bf 0}|.
\end{equation*}

From this Remark we see that $\lambda^{(k)}$ can be obtained by
performing the Hadamard test-based measurement procedure as
described in the previous Subsection with
$W = {V^{(k)}}^\dagger U_\psi$ (see Fig.~\ref{fig:updating-k-circ}(b)).
We will denote this operation as
\begin{align*}
  {\Lambda}\colon & \Hilbert^n \to [0, 1]      \\
  & \ket{w} \mapsto \lambda.
\end{align*}

By combining the classical Algorithm~\ref{alg:hopm} with the quantum
operations above and making use of one-qubit tomography
we execute the most memory-intensive operations (contractions of a
tensor of size $2^n$) of Algorithm~\ref{alg:hopm} in the quantum
domain using $n$ qubits; these steps are given in Algorithm~\ref{alg:quals}.
The main steps of the algorithm implementation (lines
  \ref{line:hopm_v_update} and \ref{line:hopm_nmode_prod} in
Algorithm~\ref{alg:quals}) are summarized in Fig.~\ref{fig:updating-k-circ}.
The result is summarised in Theorem~\ref{theorem:qhopm}.

\begin{algorithm}
  \caption{Iterative quantum implementation of Algorithm~\ref{alg:hopm}.
  }
  \label{alg:quals}
  \begin{algorithmic}[1]
    \Procedure{QHOPM}{$U_{\psi}$, $\epsilon$}
    \State{Choose $(\vartheta^{(0)}_i, \varphi^{(0)}_i) \in [0,
    \pi) \times [0, 2\pi)$  for $i=1,\ldots,n$}
    \State{initialise $k$ to $0$}
    \While{\label{line:qhopm_while}$\vert\lambda_n^{(k+1)}-\lambda_n^{(k)}\vert
    > \epsilon$}
    \For{$i \in {1,\ldots,n}$}
    \State{$\varphi_i^{(k+1)}, \vartheta_i^{(k+1)} =
      \mathcal{T}_i({V_i^{(k+1)}}^{\dagger} U_\psi \ket{\bf
    0})$}\label{line:qhopm_v_update}
    \EndFor
    \State{$\lambda^{(k+1)} = \Lambda({V^{(k+1)}}^\dagger U_\psi
    \ket{\bf 0}) $}\label{line:qhopm_nmode_prod}
    \State{increment $k$}
    \EndWhile
    \Return{$\lambda^{(k+1)}$ and $\{ (\vartheta^{(k+1)}_i,
    \varphi^{(k+1)}_i) \}_{i=1\ldots n} $}
    \EndProcedure
  \end{algorithmic}
\end{algorithm}

\begin{figure}[ht]
  \centering
  \includegraphics{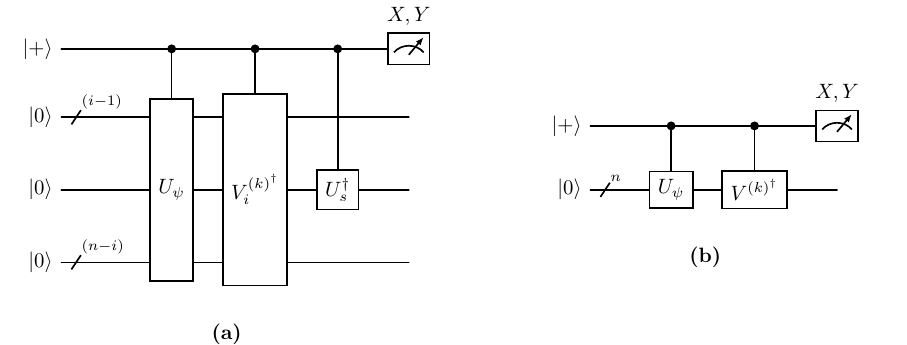}
  \caption{\textbf{Circuit representations of the main steps of
    QHOPM for the $k$-th iteration.}
    \textbf{(a)} one-qubit tomography for measuring $u_{i,s}^{(k)}$
    (see Sec.~\ref{sec:one-qubit_tomography});
    \textbf{(b)} measuring $\lambda^{(k)}$. The meter labels show the
  operators to measure.}
  \label{fig:updating-k-circ}
\end{figure}

\begin{theorem}\label{theorem:qhopm}
  Given
  a unitary operator $U_\psi \in \mathrm{U}(2^n)$, that prepares
  a target state $\ket{\psi} = U_\psi \ket{\mathbf{0}}$ and a
  sufficiently small $\epsilon$,
  QHOPM (Algorithm~\ref{alg:quals}) returns a pair
  $\lambda$
  and $\{ (\vartheta_i, \varphi_i) \}_{i=1\ldots n}$.
  Here
  $\{ (\vartheta_i, \varphi_i) \}_{i=1\ldots n}$
  are the parameters of the encoding of a separable state $\ket{\phi}$.
  These values obey the following conditions:
  \begin{enumerate}
    \item $\ket{\phi}$ is a RTA of $\ket{\psi}$.
    \item $\lambda = | \langle \phi \vert \psi \rangle|$ is
      an approximation of the entanglement eigenvalue of $\ket{\psi}$
      up to $\epsilon$.
  \end{enumerate}
\end{theorem}
\begin{proof}

  We note that HOPM (Algorithm~\ref{alg:hopm}) satisfies
  points 1 and 2 of the theorem for the components $T_\psi$ of the
  state $\ket{\psi}$ and one-qubit components $\mathbf{v}_i$ of
  $\ket{\phi}$. Thus, to prove the points 1 and 2 for QHOPM it is
  enough to prove the equivalence of
  lines~\ref{line:hopm_v_update}-\ref{line:hopm_nmode_prod} of
  Algorithm~\ref{alg:hopm} to lines~\ref{line:qhopm_v_update}
  and~\ref{line:qhopm_nmode_prod} of Algorithm~\ref{alg:quals}.

  First, let us prove the equivalence of
  lines~\ref{line:hopm_v_update} and~\ref{line:hopm_v_norm} of HOPM
  to line~\ref{line:qhopm_v_update} of QHOPM\@.
  Having the encoding of the state $\ket{\mathbf{v}_i^{(k+1)}}$ in
  terms of the angles $(\vartheta_i^{(k+1)}, \varphi_i^{(k+1)})$,
  obtained in line~\ref{line:qhopm_v_update} of QHOPM, we
  reconstruct the result of lines~\ref{line:hopm_v_update}
  and~\ref{line:hopm_v_norm} of HOPM up to a global phase as
  follows, due to Proposition~\ref{prop:alg-equiv},
  \meqref{eq:phi-encode} and Remark~\ref{prop:global-phases}:
  \begin{equation}
    v_{i,b}^{(k+1)} = \langle b | \mathbf{v}_{i}^{(k+1)} \rangle =
    \bra{b} R_z(\varphi_i^{(k+1)}) R_x (\vartheta_i^{(k+1)}) \ket{0}.
  \end{equation}
  To show the other direction, having the components
  $v_{i,b}^{(k+1)}$, obtained in lines~\ref{line:hopm_v_update}
  and~\ref{line:hopm_v_norm} of HOPM, we solve the system of
  equations~\meqref{eq:tomography-system} and recover the angles
  $(\vartheta_i^{(k+1)}, \varphi_i^{(k+1)})$. This proves the
  equivalence of
  lines~\ref{line:hopm_v_update}-\ref{line:hopm_v_norm} of HOPM to
  line~\ref{line:qhopm_v_update} of QHOPM\@.
  This result does not depend upon $i$ or $k$, thus, it is valid
  for any $i$ and $k$.

  In Remark~\ref{prop:lambda-circ-equiv} we showed how to use the
  obtained angles $(\vartheta_i^{(k+1)}, \varphi_i^{(k+1)})$ to
  calculate $\lambda^{(k+1)}$, which is just a way to rewrite
  line~\ref{line:hopm_nmode_prod} of HOPM using $U_\psi$ and
  $V^{(k+1)}$. The measurement of $\lambda^{(k+1)}$ is denoted as
  $\Lambda({V^{(k+1)}}^\dagger U_\psi \ket{\bf 0})$, which is
  line~\ref{line:qhopm_nmode_prod} of QHOPM\@. This shows the
  equivalence of line~\ref{line:hopm_nmode_prod} of HOPM to
  line~\ref{line:qhopm_nmode_prod} of QHOPM for any $k>0$ and
  finishes the proof.
\end{proof}

\begin{theorem}
    The total shot complexity of QHOPM for all iterations needed to 
    obtain the approximation of $\lambda$ with accuracy $\epsilon$ 
    having the accuracy of each measurement $\delta$ is 
    \[
        \bigO{\frac{n^3}{\epsilon \delta^2}}.
    \]
\end{theorem}
\begin{proof}
Recall that $n$ is the number of qubits in the  input unitary
$U_\psi$\footnote{One may assume the circuit description of
  $U_\psi$
is of length polynomial in $n$.}
that defines  the target state $\ket{\psi}$.
Assuming the use of the Hadamard test for the $n$-qubit target
state $\ket{\psi}$, QHOPM requires $n+1$ qubits.

Each application of the tomography step on
line~\ref{line:qhopm_v_update} of Algorithm~\ref{alg:quals}
requires $4$ measurements to recover $\ket{\mathbf{v}_i^{(k)}}$, and
each estimation of the scalar product at
line~\ref{line:qhopm_nmode_prod} of Algorithm~\ref{alg:quals}
requires $2$ measurements.
Thus, each iteration of the loop in line~\ref{line:qhopm_while}
uses $4n + 2$ measurements.
Each measurement (assuming the Hadamard test)  requires
$\bigO{\delta^{-2}}$ single-shot readouts\footnote{Note it is
  possible to reduce the amount of measurements required to
$O(\delta^{-1})$ via amplitude amplification techniques.}
due to the Chernoff bound for absolute error $\delta$.
The total shot complexity of each iteration is $\bigO{n\delta^{-2}}$.
Recall from Theorem~\ref{thrm:hopm-conv} that the number of iterations
needed for HOPM to converge up to accuracy $\epsilon$ scales 
with $n$ as $\bigO{{n^2}/{\epsilon}}$. Thus, taking into account 
Theorem~\ref{theorem:qhopm} about equivalence of HOPM and QHOPM, 
QHOPM's time complexity is
\[
    \bigO{\frac{n^3}{\epsilon \delta^2}}.
\]

\end{proof}

\section{QHOPM simulation}
\label{ssec:quantum_circuit_simulation}
\subsection{Simulation setup}
We simulate the Hadamard test variant of QHOPM using the
Qiskit~\cite{Qiskit2023} platform (IBM Qiskit Version 0.46.0).
We find in practice that the algorithm gives sufficiently good results
when each individual simulation uses $1\times10^5$ shots which
corresponds to an absolute accuracy $\epsilon \approx 0.003$.

For a given a target state, regions of initial separable states
will converge to different local minima~\cite{DeLathauwer2000}.
When using the HOPM algorithm one might try many initial separable
states and choose the minimum value obtained.
In this work we use 10 random initial states to demonstrate the
variation between initial starting points.

To investigate the effect of noise on QHOPM, we use two types of 
noise models:
1)~simplified model, which is basically a depolarising noise 
channel, applied to all the gates with the same noise rate $p$ 
using the Qiskit AerSimulator; 
2)~realistic model, implemented in ``FakeLima'' and ``FakeSherbrooke'' 
Qiskit backends, developed by IBM to mimic the real ``Lima'' and 
``Sherbrooke'' superconducting qubit IBM's devices (to ensure reproducibility).

To demonstrate the performance of QHOPM we chose the following target states:
\begin{itemize}
  \item 3 qubit W state, with known $\GE = 5/9$~\cite{Wei2003};
  \item 3, 4, 5, 6, and 9 qubit GHZ states (\GHZ{n}), each with known $\GE = 0.5$;
  \item 3, 6 and 9 qubit ring cluster states~\cite{Guehne2009} (\Ring{n});
  \item 3, 4, 5, 6 and 9 qubit random states (\Random{n}).
\end{itemize}

To generate \Random{n} states we created quantum circuits with the
following method:
we sampled gates (uniformly at random) from the gate set of CNOT
and Qiskit's $U$ gate (general single qubit rotation) with random angles
and applied them to uniformly randomly selected qubits.
This process was iterated until the circuit depth reached 10.

The initial random separable states $V^{(0)}$ (see
\meqref{eq:phi-encode}) were chosen by sampling  uniformly at random
$\vartheta_i \in [0, \pi)$ and
$\varphi_i \in [0, 2\pi)$ for $i = 1,\dots,n$.

\subsection{Analysis of quantum noise effects in QHOPM and ``proof of concept'' mitigation}

In this Section we will analyse the effects of the following noise types on QHOPM: 1) statistical noise connected to a finite number of shots; 2) SPAM noise; 3) incoherent quantum noise. We assume that the the gates are well calibrated and coherent noise effects are negligible in comparison to the noise types above.

\begin{lemma}[On statistical noise effects]
    On the $k$-th iteration of QHOPM statistical noise shifts the median value of $\lambda^{(k)}$ by $\delta \lesssim 1/\sqrt{n_s}$, where $n_s$ is the number of shots used for the measurements.
\end{lemma}
\begin{proof}
    In terms of the Hadamard test ${\lambda^{(k)}}^2 = \braket{X_a}^2 + \braket{Y_a}^2 = x_a^2 + y_a^2$.
    One can define median $m$ as a minimiser of the mean absolute error $\mathbb{E}[|X-c|]$ of a real number $c$ w.r.t. random variable $X$ over the distribution of $X$~\cite{Stroock2011}: $m = \arg\min_X \mathbb{E}[|X-c|]$. If we put $c = (\hat{\lambda}^{(k)})^2$ and $X = ({\lambda}^{(k)})^2 = ((\hat{\lambda}^{(k)})^2 + (\zeta_x^2+\zeta_y^2) + 2(\zeta_x x_a+\zeta_y y_a))$, then by taking the median over the sample results we are minimizing the expected value of the squared fluctuation $\delta = \mathbb{E}[\zeta^2]$, which according to the Hoeffding's inequality will be upper bounded by a value proportional to $1/\sqrt{n_s}$, where $n_s$ is the number of shots used for the measurements.

    Measured values $|\braket {b_{[i]} | w}|^2$ in probabilities $P_i(b,W)$ in \meqref{eqn:tomography-prob-s} will experience similar fluctuations, which lead to the noisy updates of $\ket{{\bf v}_i^{(k)}}$. This in turn leads to stochastic addition to the updates of $\ket{{\bf v}_{i+l}^{(k)}}$ with $0<l\leq(n-i-1)$, which can be added as an additional $\zeta$ fluctuation to $\lambda^{(k)}$ and analysed as given above.
\end{proof}

\begin{lemma}[On the effect of SPAM noise]
    SPAM noise does not change the value of $P_i(s)$, but introduces multiplicative error for the value of $\lambda^{(k)} \to (c_{00} - c_{10}) \lambda^{(k)}$, where $c_{ij}$ are elements of the SPAM confusion matrix $C$, which correspond to probabilities of reading out state $\ket{j} \in \{ \ket{0}, \ket{1} \}$, having the actual state $\ket{i} \in \{ \ket{0}, \ket{1} \}$.
\end{lemma}
\begin{proof}
    In terms of the Hadamard test ${\lambda^{(k)}}^2 = \braket{X_a}^2 + \braket{Y_a}^2 = x_a^2 + y_a^2$. Each $x_a$ and $y_a$ are the measurement results for probabilities of certain one-qubit states $\rho_x = H\rho H $ and $\rho_y = S^\dagger H \rho H S$ being in the state $\ket{0}$ performed on the ancilla qubit. Let the SPAM noise be defined by confusion matrix $C$. In this case, 
    \begin{equation}\label{eqn:spam-noisy-measurement}
    \begin{split}
        x_a =& c_{00} \bra{0}\rho_x\ket{0} - c_{10} \bra{1} \rho_x \ket{1}
        = c_{00} \bra{+}\rho\ket{+} - c_{10} \bra{-} \rho \ket{-} \\
        =& \frac{c_{00}}{2} (\bra{0} \rho \ket{1} + \bra{1} \rho \ket{0}) 
        - \frac{c_{10}}{2} (\bra{0} \rho \ket{1} + \bra{1} \rho \ket{0}) \\
        =& (c_{00} - c_{10}) x_a'.
    \end{split}
    \end{equation}
    Similar result we get for $y_a$, which brings the results for $\lambda^{(k)}$.

    The result for $P_i(s)$ can be obtained by noticing, that we will get the same result as in \meqref{eqn:spam-noisy-measurement} for $|\braket{b_{[i]}|w}'|^2$ and $|\braket{b^\perp_{[i]}|w}'|^2$:
    \[
        P_i(s) =
      \frac{(c_{00} - c_{10})|\braket{b_{[i]}|w}'|^2}{(c_{00} - c_{10})|\braket{b_{[i]}|w}'|^2 +
      (c_{00} - c_{10})|\braket{b^\perp_{[i]}|w}'|^2 } = P_i'(s),
    \]
    where $P_i'(s)$ is noiseless value of $P_i(s)$.
\end{proof}

As for the incoherent noise first of all we would like to note that
for different topologies of a device, the Hadamard test-based implementation of QHOPM will be compiled in a different manner, 
which will change the noise channel, especially in case of non-Markovian noise. 
However, we can approximate a general scaling of noise correction to the entanglement eigenvalue for ``small'' additive uncorrelated (between algorithm iterations) noise with zero mean. From lines~6 and~8 in Algorithm~\ref{alg:hopm} we can see that the expected value of the second-order noise corrections for $\lambda^{(k)}$ scale with the number of qubits as $\bigO{n^3}$ at most. Also note that noise in the results of the one-qubit tomography $\mathcal{T}_i$ can be assumed to be an addition to incoherent noise for the $(i+1)$th separable state mode. 
This means that to get stable results, the second-order correction terms should scale as $\bigO{n^{-3}}$. 
This makes the algorithm unusable for large qubit number NISQ devices; however, it can be used on modern devices with a moderate number of qubits if we know some information about the noise channel.

To demonstrate this we simulate QHOPM for 3 different topologies and corresponding noise realizations (see the previous Subsection for details) and show a ``proof-of-concept'' way to mitigate noise, based on the approach 
proposed in~\cite{Vovrosh2021,Urbanek2021}. 
This approach is based on the assumption that the noise in the system can be treated as depolarizing noise (DN)~\cite{NielsenChuang,Emerson2005}
\begin{equation}\label{eqn:depol_channel}
  \mathcal{E}(\rho) = (1-p)\rho + p \frac{1}{2^n} I^n,
\end{equation}
where $I$ is a $2$-dimensional identity matrix; $0 \leq p \leq
4^n/(4^n - 1)$ is a parameter of the model (which we will refer to
as the noise rate). When $p=1$ the channel is fully depolarising; when $p = 4^n/(4^n - 1)$
the channel is equivalent to random Pauli errors applied
(uniformly) on each qubit with equal probability.
The parameter $p$ can be found using the following heuristics: by measuring $\braket{\rho} = \Tr[\rho^2]$~\cite{Vovrosh2021}; by measuring the expectation value of an observable with known result~\cite{Urbanek2021}. Our approach is similar to the one used in~\cite{Urbanek2021}, and is described further in the text. 

In general, the assumption of the global noise channel being depolarizing noise is not realistic. This is a significant underestimation of a noise model, which also does not include coherent noise effects and possible non-Markovian behaviour. Nevertheless, this heuristic mitigation approach has shown itself to work reasonably well for real devices~\cite{Vovrosh2021, Urbanek2021}. We note that in~\cite{Urbanek2021} the authors have also used a randomized compiling procedure~\cite{Wallman2016} to convert coherent noise effects into incoherent. In our case, we use fake IBM backends (see the previous Subsection), which to our knowledge do not include coherent noise.

As our work was not aimed at detailed noise analysis for a particular device, we
consider this simplified model of noise mitigation and analyze performance of QHOPM with respect to it.
We assume that the DN channel is applied $d$ times at a constant
rate $p$ to \emph{all qubits} %
after each layer of gates,
where $d$ is the circuit depth (maximum number of operations on a
qubit in the circuit).
The DN channel commutes with any unitary operation:
\[
U\, \mathcal{E}(\rho) \, U^\dagger
=  U \left( (1-p)\rho  + p \frac{I}{2^n} \right) U^\dagger
= \mathcal{E}(U \rho U^\dagger) ,
\]
which allows to apply the channel $d$ times at the end of the circuit:
\begin{equation}\label{eqn:depol_dm}
  \begin{split}
    \rho = \mathcal{E}^{d} (\rho') & = (1-p)^d \rho' +
    \frac{p}{2^n}\sum_{i=1}^d (1-p)^{d-i} I \\
    & = q^d \rho' +\frac{1-q^d}{2^n} I,
  \end{split}
\end{equation}
where $q = (1-p)$; $\rho'$ is a density matrix of a pure state with no noise.
The number $d$ is also be understood as the number of times that
the DN channel affects the state with the error rate $p$.

\begin{proposition}\label{prop:noise}
  Given the DN channel defined by \meqref{eqn:depol_dm},
  in terms of the Hadamard test procedure and sufficient number of shots,
  the measured $P_i(s)$ ($P_i(s, U)$ for some arbitrary $U$) and
  $\lambda^{(k)}$ can be expressed as
  \begin{equation}\label{eqn:P_i(s)_noisy}
    P_i(s) = \frac{P_i'(s)}{1-\eta P_i'(s^\perp)} \approx
    \left( 1 + \eta P_i'(s^\perp) + \ldots \right)\, P_i'(s);
  \end{equation}
  \begin{equation}\label{eqn:lambda_k_noisy_hadamard}
    \lambda^{(k)} = q^{d} 
    \lambda'^{(k)},
  \end{equation}
  where %
  $P_i'(s)$, $P_i'(s^\perp)$ and $\lambda'^{(k)}$ are the
  noise-free values of $P_i(s)$, $P_i(s^\perp)$ and
  $\lambda^{(k)}$, respectively;
  dots in \meqref{eqn:P_i(s)_noisy} correspond to the terms of the
  second order and higher in $\eta$.
\end{proposition}
\begin{proof}
  To obtain $P_i(s)$, defined by \meqref{eqn:tomography-prob-s}, we
  need to measure $|\braket{b_{[i]}|w}|$, which in terms of the
  Hadamard test procedure is done by measuring
  \begin{equation}\label{eqn:average-ancilla-xy}
    \begin{split}
      & \langle X_a \rangle = 2P^{(a)}(+) - 1, \\
      & \langle Y_a \rangle= 2P^{(a)}(\imi) - 1,
    \end{split}
  \end{equation}
  where
  \begin{equation}\label{eqn:noisy_ancilla_probs}
    \begin{split}
      P^{(a)}(+) =  & \Tr(\ket{+}\bra{+}\otimes I^{n} \rho) = q^d
      {P^{(a)}}'(+) + \frac{1-q^d}{2},                    \\
      P^{(a)}(\imi) = & \Tr(S^\dagger\ket{+}\bra{+}S \otimes I^{n}
      \rho) = q^{d} {P^{(a)}}'(-i) + \frac{1-q^{d}}{2}
    \end{split}
  \end{equation}
  are the ``noisy'' probabilities of the ancilla qubit being in the
  state $\ket{+}$ and $\ket{\imi}$, respectively, and ${P^{(a)}}'(+)$
  and ${P^{(a)}}'(\imi)$ are their ``non-noisy'' counterparts; 
  we have assumed that the phase gate is noiseless, which is usually the case.
  After substituting \meqref{eqn:noisy_ancilla_probs} to
  \meqref{eqn:average-ancilla-xy}, for $u_{i,s}^{(k)}$ we have:
  \begin{equation}\label{eqn:noisy-u_is}
    |\braket{b_{[i]}|w}|^2 = \langle X_a \rangle^2 + \langle Y_a
    \rangle^2 = q^{2d} 
    |\braket{b_{[i]}|w}'|^2,
  \end{equation}
  where $\braket{b_{[i]}|w}'$ is a non-noisy value of $\braket{b_{[i]}|w}$;
  Similarly, we get
  \begin{equation}\label{eqn:noisy-u_is-perp}
    |\braket{b^\perp_{[i]}|w}|^2 = q^{2d+2} 
    |\braket{b^\perp_{[i]}|w}'|^2.
  \end{equation}

  Substituting \meqref{eqn:noisy-u_is} and
  \meqref{eqn:noisy-u_is-perp} to \meqref{eqn:tomography-prob-s},
  we have for $P_i(s)$:
  \begin{equation}
    \begin{split}
      P_i(s) &=
      \frac{|\braket{b_{[i]}|w}'|^2}{|\braket{b_{[i]}|w}'|^2 +
      |\braket{b^\perp_{[i]}|w}'|^2 - \eta |\braket{b^\perp_{[i]}|w}'|^2} \\
      &= \frac{P'_i(s)}{1-\eta P'_i(s^\perp)}
      \approx  \left(1 + \eta P'_i(s^\perp) + \ldots \right) P'_i(s),
    \end{split}
  \end{equation}
  where $\eta = (1-q^2)>0$; the right-hand side is the Taylor
  expansion of the left side near $\eta=0$; $P'_i(s^\perp)$ is a
  non-noisy value of $P_i(s^\perp)$. This proves the first
  statement of the Proposition.

  Following the same steps for $\lambda^{(k)}$, assuming that the
  noise does not affect one-qubit tomography results, we have in
  terms of the Hadamard test procedure:
  \begin{equation}
    \begin{split}
      \lambda^{{(k)}^2} & = \langle X_a \rangle^2 + \langle Y_a
      \rangle^2 = q^{2d} 
      \lambda'^{{(k)}^2},
    \end{split}
  \end{equation}
  where 
  $\lambda'^{{(k)}}$ is a non-noisy
  entanglement eigenvalue. This proves the second statement of the
  Proposition.
\end{proof}

If all the $\eta$ terms in \meqref{eqn:P_i(s)_noisy} are negligible
($p \ll 1$), the one-qubit tomography procedure $\tomo{i}$ will
return angles with errors of similar magnitude for any circuit
depth. If these errors can be neglected, the convergence condition
of QHOPM, given in line~\ref{line:qhopm_while} in
Algorithm~\ref{alg:quals}, will change only by the prefactor
defined in \meqref{eqn:lambda_k_noisy_hadamard}.
In this work, the largest error rate that we have studied is
$p=0.05$ which gives $\eta \approx 0.1$.
Thus, according to Proposition~\ref{prop:noise} and
assuming $P_i'(s) = P_i'(s^\perp) = 0.5$,
$P_i(s)$ changes at most by approximately $0.05$. This change does
not affect the convergence, as it is seen from the simulations.
We neglect this effect in what follows and take into account only
the changes of $\lambda^{{(k)}}$,
which even for relatively shallow $U_\psi$ may
be significant.
With these assumptions we approximate the true $\GE'$ as follows:
\begin{equation}\label{eqn:noisy_gm}
  \GE' = 1 - \frac{1-\GE}{(1-p)^{2d}}.
\end{equation}
The
denominator 
is always less (or equal) than one,
which means that $\GE \geq \GE'$ within our assumptions, which
corresponds to the observed simulation results.

To mitigate the results of our algorithm on hardware or simulations where the
noise model is usually unknown, we use the following procedure.
We run the algorithm on a reference state with a known $\trueGE$
and measure the noisy value $\GE$.
In our case we chose \TState{GHZ}{n} with $\trueGE = 0.50$ since
the value is consistent for any number of qubits.
Using this value and \meqref{eqn:noisy_gm} we calculate $p$ as follows:
\begin{equation}\label{eqn:real-noise-p}
  p = 1 - \left(\frac{1 - \AVG}{1 - \trueGE}\right)^{1/2d}.
\end{equation}

The mitigation procedure now proceeds as outlined in the
depolarising noise case.

\subsection{Simulation Results}
In this section, we present simulation results of our algorithm evaluating its robustness to noise
using the IBM Qiskit~\cite{Qiskit2023} platform. We also show how to mitigate the effects of noise on the algorithm.
We refer to target states by the notation \TState{State}{n} for $n$
qubits, for example
\GHZ{9} is the GHZ state with 9 qubits
and \Random{3} is a particular random state with 3 qubits.

When using HOPM we choose the minimum $\GE$ from a sample of
initial separable states, however (due to the finite number of
shots) QHOPM tends to fluctuate around the convergence value.
We find that the median gives a more accurate summary of the convergence
value of QHOPM than the mean, perhaps reflecting some skew in the
distribution.

For a given target state
let $\setEGk$ represent a set of estimations of $\GE$ for a sample
of initial separable states at iteration $0 \leq k \leq K$.
Let $\AVG$ represent the value converged to by QHOPM, computed as
the median of the values in $\setEG = \{ \medSetEGk \mid k \in
\left[K-5,\ldots K\right]\}$,
where $\med{\mathcal{A}}$ is the median of the values in the set
$\mathcal{A}$.
Let $\STD$ be the interquartile range for $\setEG$.
We note that the choice of 5 final iterations here is arbitrary
(statistically better results are obtained if one uses more iterations).
Here we choose the median instead of the mean to avoid possible
outliers induced by certain initial separable states.

We begin by verifying that our quantum implementation matches the classical
algorithm and demonstrate the effect
on convergence of the depolarising noise model
with noise rate $p=0.01$.
We use the \GHZ{9} state, which has a known $\trueGE = 0.50$, as an example.
In Fig.~\ref{fig:error_with_shots_noise} we see a good
correspondence with classical HOPM for $1\times 10^{7}$ shots per
measurement (yielding an expected error of $\epsilon \approx 3.2
\times 10^{-4}$ according to the Chernoff bound), that is QHOPM
converges to the same value ($\AVG \approx 0.50$, $\STD \lesssim
10^{-4}$) as classical HOPM ($\AVG = \trueGE = 0.50$).
Then we reduced the number of shots to $1\times 10^5$ ($\epsilon
\approx 3.2 \times 10^{-3}$) and
we see convergence to $\AVG \approx 0.503$, within the slightly
higher $\STD \approx 0.005$ of the true value.
These results are consistent with the Chernoff bound discussed in
the previous Section.
Finally, we introduced depolarising noise  with
error rate $p=0.01$.
Remarkably, the quantum algorithm still managed to converge with
little variability ($\STD \approx 0.003$),
however, the value it converged to ($\AVG\approx 0.679$) is greater
than the true value.

\begin{figure}[h]
  \centering
  \includegraphics{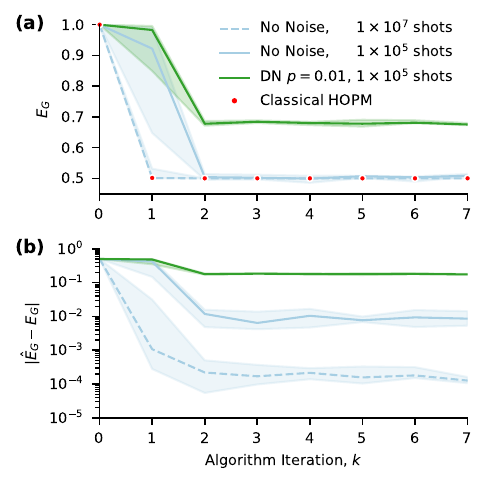}
  \caption{\textbf{The results of HOPM and QHOPM
      for (a) approximating the geometric entanglement \GE{} of
    \GHZ{9} and (b) the absolute error ($|\trueGE - \GE|$).}
    The lines represent the median value for 10 random initial
    separable states for different numbers of shots per measurement
    and noise models (see legend).
    The shaded areas show the interquartile ranges.
  }
  \label{fig:error_with_shots_noise}
\end{figure}

Notice that the distribution of $\GE$ values found after 1
iteration by QHOPM is different from that found by HOPM for the
same initial separable states.
By the second iteration, HOPM and QHOPM converge to the same values.
This divergence (from our experience) is normally not so
noticeable, but of all the states we have studied, this effect is most pronounced in \GHZ{9}.
We assume that resolution errors, due to the reduced number of
shots in the measurements, are amplified cumulatively for certain
states (\GHZ{9}, in particular) during one-qubit tomography
resulting in a different (but still improved) separable state from
what HOPM or QHOPM with infinite resolution would find.
From this point on we proceed with $1\times 10^5$ shots and the
median over 10 initial separable states for all simulations (unless
said otherwise) since it gives a sufficiently good convergence
w.r.t. classical HOPM\@.

To evaluate the performance of QHOPM for random target states we
ran QHOPM with no noise for 100 \Random{n} target states
each with 20 initial separable states
with $n\in\{3,\ldots,6\}$ qubits.
In Fig.~\ref{fig:all_qubits_random_error} we summarize the results
of the simulation.
Here for each target state at iteration $k$ we calculate
the median result $\medSetEGk$ over 20 initial separable states and
plot the median absolute difference $| \trueGE - \medSetEGk|$ (and
the interquartile range) over all the target states for each $k$
and number of qubits $n$.
We show that the median
value of the absolute error
$|\trueGE - \medSetEGk|$ converges to the values
less than $10^{-2}$ for any number of qubits, which is again
consistent with the Chernoff bound.
In the case of negligible noise, the accuracy of QHOPM depends only
on the number of shots per measurement.

\begin{figure}
  \centering
  \includegraphics{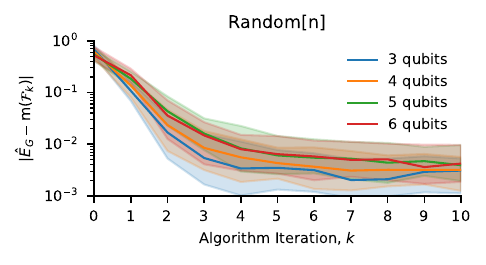}
  \caption{\textbf{The absolute error ($|\trueGE - \medSetEGk|$) of
    QHOPM approximating \GE{} of \Random{n} states with no noise.}
    The lines represent the median absolute error over 100
    \Random{n} target states ($n \in \{3,\ldots,6\}$ qubits)
    between the best (minimal) HOPM results and the median results
    $\medSetEGk$ computed by QHOPM for each target state over 20
    random initial separable states. The shaded areas show the
    interquartile ranges for the absolute error.
  }
  \label{fig:all_qubits_random_error}
\end{figure}

To explore the relationship between DN error rate $p$ and the divergence of
QHOPM's \GE{} value from HOPM's, we ran simulations with $p \in
\{0.001, 0.01, 0.05\}$ representing ``acceptable'', ``bad'' and
``terrible'' levels of noise, respectively.
We used \GHZ{9} and \Random{6} as the target states.
In Fig.~\ref{fig:mitigation}(a) we see that for \GHZ{9} when
$p=0.001$ QHOPM converges to  $\AVG \approx 0.519, \STD \approx  0.003$.
The upper limit of error rate $p$ that is considered acceptable at
this early stage of quantum computing is $0.01$,
at this level, the noise is causing our algorithm to converge to
the value $\AVG \approx 0.679$ with $\STD \approx 0.003$.
At $p = 0.05$ the algorithm still converges with a similar
variability but to the incorrect value $\AVG \approx 0.946$ with
$\STD \approx 0.002$.
The results are also summarised in Table~\ref{tab:mitigation}.

\begin{figure*}[ht]
  \centering
  \makebox[\textwidth][c]{
    \includegraphics{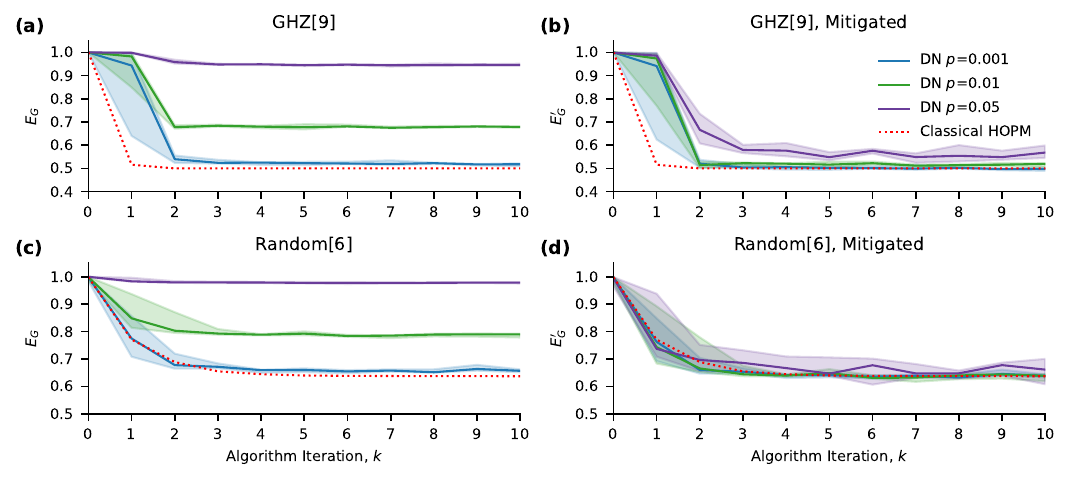}
  }%
  \caption{\textbf{Simulation results for QHOPM on target states
    \GHZ{9} and \Random{6} with the DN model and its mitigation.}
    \textbf{(a)} and \textbf{(c)} show the simulation results for
    QHOPM with input target states \GHZ{9} and \Random{6},
    respectively, for different DN error rates;
    \textbf{(b)} and \textbf{(d)} show the results of our
    mitigation procedure (see \meqref{eqn:noisy_gm}) applied to
    results in \textbf{(a)} and \textbf{(c)},
    respectively.
    Solid lines represent the median geometric entanglement,
    $\GE{}$, over the results for different initial separable states and the
    shaded colours represent the interquartile range.
    The colour of the lines represents the DN noise rate used for the
    simulation (see legend).
    Red dotted lines show the results of classical HOPM\@.
  }
  \label{fig:mitigation}
\end{figure*}

\begin{table}
  \caption{Results of HOPM and QHOPM with different DN noise rates
    $p$ for \GHZ{9} and \Random{6} states (see.
    Fig.~\ref{fig:mitigation}). Minimum HOPM estimate for \Random{6}
  is $\GE=0.636$. }
  \label{tab:mitigation}
  \begin{tabular}{lcllllrrcc}
    \toprule
    Target  & DN $p$  & \multicolumn{2}{c}{Measured} &
    \multicolumn{2}{c}{Mitigated} \\
    &  &    $\AVG$ & $\STD$ & $\AVGm$ & $\STDm$ \\
    \midrule
    \GHZ{9}                    & 0.001  & 0.519 & 0.003 & 0.499 & 0.003\\
    & 0.010  & 0.679 & 0.003 & 0.516 & 0.006\\
    & 0.050  & 0.946 & 0.002 & 0.554 & 0.019\\
    \Random{6}                 & 0.001  & 0.657 & 0.003 & 0.638 & 0.003\\
    & 0.010  & 0.789 & 0.004 & 0.639 & 0.008\\
    & 0.050  & 0.978 & 0.000 & 0.661 & 0.029\\
    \bottomrule
  \end{tabular}
\end{table}

In Fig.~\ref{fig:mitigation}(c) we give the results for the median
$\GE$ for \Random{6}.
With low noise ($p = 0.001$) QHOPM converges to $\AVG \approx
0.657, \STD\approx 0.003$, which is within $\epsilon \approx 0.02$
w.r.t. the classical algorithm's result.
Again, we see that as noise increases, the converged value
increases: for $p=0.01$ it goes to $\AVG \approx 0.789$ with $\STD
\approx 0.004$,
and for $p =0.05$ it goes to $\AVG \approx 0.978$ with $\STD\approx 0.000$.
We note that the interquartile range of the converged values is
within our measurement tolerance $\epsilon \approx 0.003$ (defined
by the chosen number of shots).

We see the same effects for: \GHZ{3}, \GHZ{6} (Supplemental
Figure~\ref{supfig:all_qubits_with_DN_mitigation_GHZ}(a),(c));
\Random{3}, \Random{9} (Supplemental
Figure~\ref{supfig:all_qubits_with_DN_mitigation_random}(a),(e));
the well studied \W{3} state with $\trueGE = 5/9$ (Supplemental
Figure~\ref{supfig:all_qubits_with_DN_mitigation_W}(a));
and the important ring cluster states \Ring{3}, \Ring{6}, \Ring{9},
(Supplemental
Figure~\ref{supfig:all_qubits_with_DN_mitigation_Ring}(a),(c),(e)).
We also observe that the shift, induced by noise, increases as the
number of qubits increases.

We then characterised the shift in $\AVG$ induced by the noise as
a function of circuit depth and noise rate.
The resulting model, based on the simple assumption
that the depolarising channel is applied with a constant
error rate $p$ after each of $d$ layers of gates, yields a formula
\meqref{eqn:noisy_gm} that
approximates the mitigated \GM{}, $\GE'$, expected without noise.
Similarly to $\AVG$ we define $\AVGm$ as the mitigated result of
QHOPM, computed as the median of the values in $\setEGm = \{
\medSetEGmk \mid k \in \left[K-5,\ldots K\right]\}$, where $\setEGmk$
represents
a set of mitigated values $\GEm$ for a sample of initial separable
states at iteration $0 \leq k \leq K$.
Similarly $\STDm$ is the interquartile range of $\setEGm$.
We observe in Fig.~\ref{fig:mitigation}(b) the effect of the
mitigation on \GHZ{9}.
For $p=0.001$, the mitigated value $\AVGm \approx  0.499$ with
$\STDm=0.003$, which corresponds to the true value.
However, for $p=0.01$, $\AVGm \approx  0.516$ with $\STDm=0.006$ is
the result.
For $p=0.05$ mitigation is still helpful and adjusts the noisy
value to $\AVGm \approx 0.554$ ($\STDm=0.019$).
We see in Fig.~\ref{fig:mitigation}(d) that the mitigation for the
state \Random{6} gives better results: for
$p=0.001$ \GE{} being mitigated to $\AVGm \approx  0.638$ ($\STDm=0.003$), for
$p=0.01 $ to $\AVGm \approx  0.639$ ($\STDm=0.008$), and for
$p=0.05 $ to $\AVGm \approx  0.661$ ($\STDm=0.029$).
We suspect that the observed improvement in mitigation is due to
the increased gate density of the \Random{6} circuit compared to \GHZ{9}.
Our model assumes the noise channel is applied to all qubits $d$
times (for circuit depth $d$) with the same error rate $p$.
However, the DN implementation in Qiskit applies the channel only
when there is a gate applied to the qubit.
The structure of the circuit for \GHZ{9} is one of the worst cases
for our mitigation method
since it includes a chain of 8 CNOTs traversing the 9 qubits one by
one, leading to a deep but sparse circuit structure.

Finally, we apply our mitigation technique to results obtained with
more realistic noise models.
We consider the IBM Qiskit ``FakeLima'' (representing an older 5 qubit device)
and ``FakeSherbrooke'' (representing a newer device) Qiskit
backends that are calibrated to match the error profile of
physical devices.
To apply our mitigation technique we need to estimate the noise
rate $p$ for DN that would
approximate the quantum device's noise.
To do this we apply QHOPM to the GHZ state, where the true
value~\trueGE{} is known.
We then use the ratio between the measured \AVG{} and the true
value \trueGE{} to estimate $p$ (see \meqref{eqn:real-noise-p}).

In Fig.~\ref{fig:mitigate_other_noise_models-ghz9-realistic}(a) we show
simulation results for QHOPM on \GHZ{9} state for the mentioned realistic
noise models and (b) the mitigated results, having the noise parameter estimated
to be $p=0.003$ for $d=20$.
In Fig.~\ref{fig:mitigate_other_noise_models} we show the median
absolute error of the QHOPM results for \Random{6} state within (a)
FakeSherbrooke and (c) FakeLima backends, and their mitigation
results (b) and (d), respectively.

\begin{figure*}[ht]
    \centering
    \includegraphics[width=\textwidth]{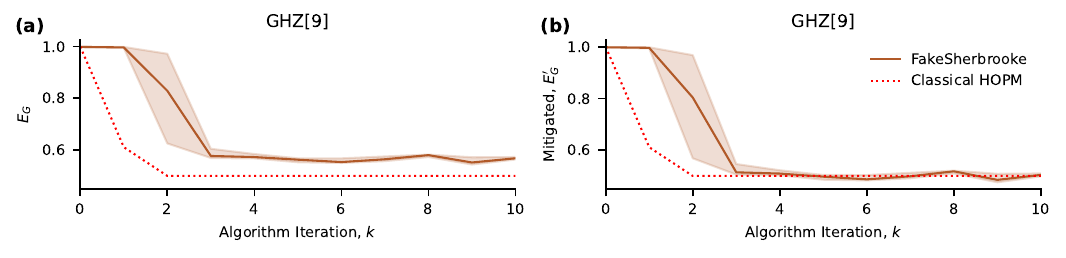}
    \caption{\textbf{Simulation results (a) for QHOPM on \GHZ{9} target state
    realistic IBM Qiskit FakeSherbrooke backend ($\AVG \approx 0.564$, 
    $\STD \approx 0.011$), and its mitigation (b) with $p \approx 0.003$ 
    ($\AVG \approx 0.498$, $\STD \approx 0.011$).}
    Solid lines represent the median geometric entanglement,
    $\GE{}$, over the results for different initial separable states and the
    shaded colours represent the interquartile range.
    The colour of the lines represents the noise model used for the
    simulation (see legend).
    Red dotted lines show the results of classical HOPM\@.}
  \label{fig:mitigate_other_noise_models-ghz9-realistic}
\end{figure*}

\begin{figure*}[ht]
  \centering
  \makebox[\textwidth][c]{
    \includegraphics{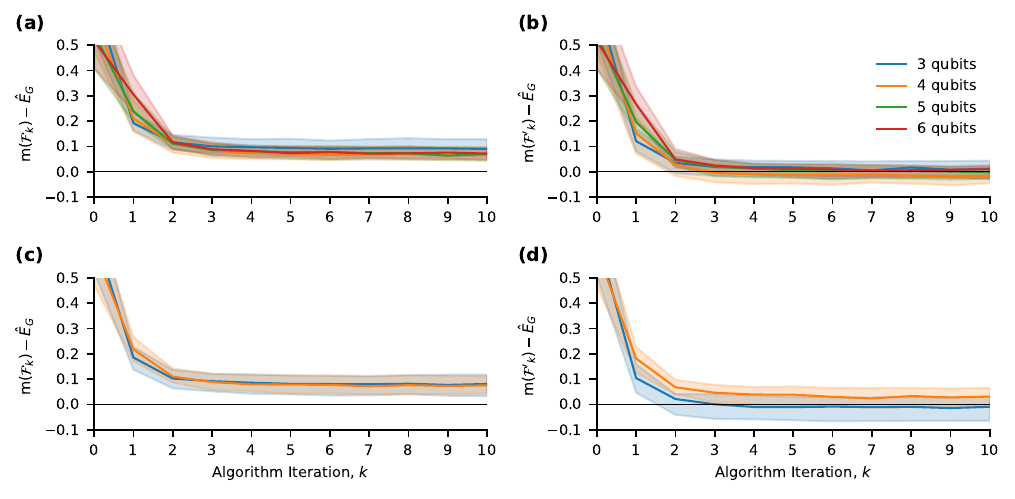}
  }
  \caption{
    \textbf{The error ($\medSetEGk - \trueGE$) of QHOPM
      approximating \GE{} for 100 different samples of \Random{n}
    with realistic noise models and their mitigation.}
    \textbf{(a)} and \textbf{(b)} show unmitigated and mitigated
    results, respectively, for \Random{n} with $n \in \{3,4,5,6\}$,
    simulated with the FakeSherbrooke Qiskit backend.
    \textbf{(c)} and \textbf{(d)} show the same for the FakeLima
    Qiskit backend ($n \in \{3,4\}$ since FakeLima imitates a 5 qubit device).
    Solid lines represent the median error over the initial
    separable states and the
    shaded colours represent the interquartile range.
    The colour of the lines represents the number of qubits (see legend).
  }
  \label{fig:mitigate_other_noise_models}
\end{figure*}

We also applied our mitigation technique to other GHZ states
(Supplemental Fig.~\ref{supfig:all_qubits_with_DN_mitigation_GHZ}(b),(d),(f)),
the W state (Supplemental
Figure~\ref{supfig:all_qubits_with_DN_mitigation_W}(b)),
some Cluster Ring states (Supplemental
Figure~\ref{supfig:all_qubits_with_DN_mitigation_Ring}(b),(d),(f)),
and some Random states (Supplemental
Figure~\ref{supfig:all_qubits_with_DN_mitigation_random}(b),(d),(f)).

\section{Discussion}
\label{sec:discussion}
In this work we presented QHOPM, an iterative quantum algorithm for
the approximation of
the geometric measure of entanglement of pure states on near term
quantum devices as well as
a simple noise mitigation technique to improve its experimental resolution.
QHOPM is the quantum adaptation of the well-known classical HOPM
approach~\cite{Shimoni2005, Most2010, Streltsov2011, Ni2014,
Hu2016, Zhang2020a} for rank-1 tensor approximation.
The convergence rate of QHOPM is given in Theorem~\ref{thrm:hopm-conv}
and it is based on the complex extension (see 
Appendix~\ref{app:hopm-convergence-proof}) of the
known HOPM convergence rate for real valued tensors~\cite{uschmajew2015new,Hu2018}.

While QHOPM does not converge differently than classical HOPM and
gives similar results (assuming a sufficient amount of shots and
relatively low noise rate) its potential utility lies in the
observation that it requires only $n+1$ qubits to compute the \GM{}
of a quantum state (on $n$ qubits) while classical HOPM naively
requires an exponential data structure (and so exponential runtime)
to work with quantum states.
We also show that the number of measurements required per algorithm
iteration scales linearly with $n$. In our work, to avoid the
problem of post-selection we have assumed Hadamard test-based
measurements, which have recently been successfully implemented on
real devices~\cite{Yoshioka2024}.

We have provided an analysis of the effect of statistical and SPAM noise
on QHOPM, however the analysis of incoherent noise is highly dependent on
the target state and the device the algorithm is executed on, especially for
possible non-Markovian effects. 
Because of this we have provided estimations on the scaling 
of the incoherent noise correction to QHOPM's results with 
the number of qubits in the second order approximation, which is
not tied to a particular device.
In simulations, we observed
that the effect of depolarising noise on QHOPM for the studied
noise rates ($p \leq 0.05$) and number of shots ($10^5$)
does not affect its convergence, but rather shifts the final result.
The same observation was made in simulations using more realistic
noise models implemented in IBM Qiskit's FakeLima and
FakeSherbrooke simulation backends.
To show potential for noise mitigation we have
successfully applied a ``proof-of-concept'' mitigation
technique, based on the work~\cite{Vovrosh2021}, which uses the 
same assumption of global depolarizing noise.
Similar approach was proposed in~\cite{Urbanek2021}, and both 
have shown reasonable results on real devices.
We have used this mitigation technique and analysed QHOPM's performance in this setting.
The noise rate for the unknown models was estimated using \GHZ{n}
as a reference state with a known \GM.
We note, that this mitigation approach is consistent if, say, each layer of the 
circuit (in terms of the Hadamard test-based measurement scheme) 
is twirled over Haar distribution (using, say, Clifford
2-design)~\cite{Dankert2009, Emerson2005}.
We expect that this mitigation procedure will not scale well due to the accumulation
of correlated errors and highly structured noise channels (e.g.~leakage noise, non-Markovian effects). 
The mitigation procedure in the studied realistic noise scenario for \GHZ{n} and \Random{n} states for $n=3,6,9$ has shown a significant improvement for the \GE, however for \Ring{n} and \W{3} states even at low qubit number the mitigation method struggled to give a reasonably close results (see Appendix~\ref{app:supplementary-sims}).
This shows that noise mitigation in principle can improve QHOPM's results, but one should use more sophisticated mitigation techniques to run it on NISQ devices.

To our knowledge, the only other quantum algorithms for calculating
geometric entanglement are variational algorithms where the ansatz
is a separable state~\cite{Zambrano2024, MunozMoller2022, Consiglio2022}.
We now compare to QHOPM the iVDGE algorithm~\cite{Zambrano2024} and
its predecessor VDGE~\cite{MunozMoller2022}. 
iVDGE is an improved version of VDGE approach changed 
for better barren plateau scaling.
In~\cite{MunozMoller2022, Zambrano2024} the depth of VDGE and iVDGE is constant, which is calculated excluding the depth of $U_\psi$. At the same time the measurement accuracy overhead scales exponentially with the number of qubits.
QHOPM utilises the Hadamard test procedure to avoid this overhead, sacrificing the depth, which (excluding $U_\psi$) is $\bigO{n}$ for an $n$-qubit state.
The number of gates for both VDGE/iVDGE and  QHOPM scales linearly with the number of qubits (again excluding $U_\psi$).
The authors of iVDGE have derived bounds on the variance of the cost function gradient ${\rm Var}[\partial I(\boldsymbol{\theta})/\partial \theta_i]$, which scales as $\bigO{n^{-2}}$, and for VDGE
as $\bigO{c^n}$, where $0<c<1$\footnote{In the derivations, the authors of~\cite{Zambrano2024} did not take into account the effects of statistical noise.}.
It can be shown (see Appendix~\ref{app:ivdge-scaling}) that
in the limits of large qubit numbers $n \gg 1$ and large iterations $k\gg 1$ the variance of the cost function gradient in terms of VDGE/iVDGE is approximately the average difference of the cost function update, which in QHOPM's case is $\lambda^{(k)} - \lambda^{(k-1)}$.
Using \meqref{eqn:hopm-compl-conv-rate-qubits} it is easy to show, that $\lambda^{(k)} - \lambda^{(k-1)}$ scales as $\bigO{n^2/k^2}$.
Comparing this with iVDGE and VDGE results for the variance, we see that HOPM 
and consequently QHOPM due to Theorem~\ref{theorem:qhopm}) scales better 
with the number of qubits. This can be confirmed by comparing iVDGE and VDGE
simulation results with ours.
iVDGE and VDGE converge more slowly than QHOPM, needing on the order of
$10^2$ iterations compared to QHOPM's $10^0$, with noise or without.
In a noiseless simulation (for Haar random target states)
with 8192 shots per infidelity evaluation
for 3 and 4 qubits, the median absolute error of iVDGE's best results for
each target state converges to approximately $10^{-4}$ (see Fig.~2
in~\cite{Zambrano2024}).
As the number of qubits increases to six,
the error for iVDGE and VDGE grows to approximately $10^{-2}$.
In similar simulations with noiseless QHOPM and $10^5$ shots,
the accuracy of the results seems less dependent on the number of
qubits and more on the number of shots, satisfying the Chernoff bound
(see Fig.~\ref{fig:all_qubits_random_error}).
So in the noiseless case, while iVDGE is more accurate with fewer
shots per measurement on small numbers of qubits, QHOPM can improve
its accuracy by increasing the number of shots w.r.t. Chernoff bound.
The authors of VDGE, use an optimizer (complex simultaneous
perturbation stochastic approximation) that is claimed to be robust
against noise.
We note that, in the presence of noise, QHOPM still converges to a
value and simple mitigation
techniques recover a number closer to the noiseless value.
When realistic noise is introduced, the accuracy of QHOPM's results 
for FakeLima IBM Qiskit backend
(see Fig.~\ref{fig:mitigate_other_noise_models}) are of the order
of VDGE results (see Fig.~6 in~\cite{MunozMoller2022}) for 
\texttt{ibm\_lima} IBM device.
The authors of iVDGE has also presented results for \GHZ{7} state
on the \texttt{ibm\_oslo} machine (see Fig.~6~\cite{Zambrano2024}), 
showing $\GE = 0.519 \pm 0.01$ after 300 iterations.
We could not use the FakeOslo backend (which simulates the \texttt{ibm\_oslo} device) 
 for \GHZ{7} since QHOPM would require 8 qubits for this problem but the \texttt{ibm\_oslo} device 
is limited to 7 qubits.
Taking into account that the given error rates for FakeSherbrooke
are similar to the error rates stated for the FakeOslo backend, and that 
the topologies of these devices are similar (\texttt{ibm\_oslo} -- H-shape; \texttt{ibm\_sherbrooke} -- heavy-hex), we can assume that results obtained on these devices are comparable.
\begin{figure*}[ht]
    \centering
    \includegraphics[width=\textwidth]{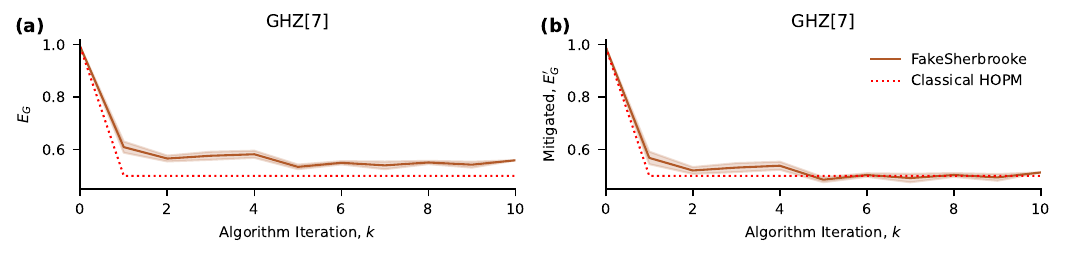}
    \caption{\textbf{Simulation results (a) for QHOPM on \GHZ{7} target state
    realistic IBM Qiskit FakeSherbrooke backend ($\AVG \approx 0.546$, $\STD \approx 0.01$), and its mitigation (b) with $p\approx 0.003$ ($\AVG \approx 0.499$, $\STD \approx 0.011$).}
    All notations and the simulation setup are the same as in Fig.~\ref{fig:mitigate_other_noise_models-ghz9-realistic}.
    }
  \label{fig:mitigate_other_noise_models-ghz7-realistic}
\end{figure*}
For the FakeSherbrooke IBM Qiskit backend
for \GHZ{7} QHOPM after 4--5 iterations approximates
$\GE = 0.546 \pm 0.01$ (see Fig.~\ref{fig:mitigate_other_noise_models-ghz7-realistic}(a)).
QHOPM converges much faster, but clearly gives worse $\GE$ estimations
than iVDGE\@.
After applying the chosen mitigation technique we were able
to get the correct values of $\GE$ within the $\STDm$ (see Fig.~\ref{fig:mitigate_other_noise_models-ghz7-realistic}(b)).
Like QHOPM, these variational approaches are sensitive to the
initial separable state, and one is never sure if the result
obtained is a local minimum or the global minimum. 

We observe that HOPM may be extended with a shift parameter to a
so-called shifted HOPM~\cite{Kolda2011, Hu2016} or Gauss-Seidel
method~\cite{Zhang2020a}.
In the Appendix~\ref{app:gauss-siedel} we show how we incorporate this
change in our algorithm.

An interesting future direction is to extend our method to
approximating geometric entanglement of mixed states~\cite{Wei2003}.
For example, the authors of~\cite{Streltsov2011} presented a method
for approximating \GM{} of a mixed state based on the connection to
the so-called revised \GM~\cite{Streltsov2010} and Uhlmann's
theorem~\cite{NielsenChuang}.
The classical implementation of the method given
in~\cite{Streltsov2011} may be naively modified to incorporate
techniques presented in the current work for some intermediate
steps, but the calculation will still be performed largely on the
classical side. Further work is needed to effectively incorporate
QHOPM with this method.

Finally, we observe that our algorithm is a quantum application of
the alternating least squares (ALS) method applied to the
particular system of non-linear equations for estimating rank-1
tensor approximation. The ALS method by itself is an area of
interest in classical machine learning~\cite{Zhou2008,Xu2019},
perhaps a variant of our approach would be of use in a suitable
problem domain.

\section*{Data Availability}
The code and data sets produced as part of this study are publicly available~\cite{zenodo_data}.

\section*{Acknowledgements}
Equal1 was funded by DTIF: QCoIr (Enterprise Ireland Project
Number: 166669/RR), and
EIC project QUENML Quantum-enhanced Machine Learning, 190105118.
Simone Patscheider was partly funded by the University of Trento
Masters program. Alessandra Bernardi was partially funded by GNSAGA
of INDAM, TensorDec Laboratory and by the European Union under
NextGenerationEU. PRIN 2022 Prot. n. 2022ZRRL4C$\_$004.
We also thank Nikolaos Petropoulos for useful discussions.

A.S. and N.M. developed the algorithm and wrote the text and simulations.
S.P. did the initial work on geometric entanglement and HOPM, and
wrote the initial draft text.
A. B. supervised S.P. and edited the manuscript text.
E. B. conceptualized the topic and edited the manuscript.

\bibliography{entanglement}

\newpage
\appendix

\section{Proof of HOPM sublinear convergence rate for complex-valued tensors}
\label{app:hopm-convergence-proof}

To obtain the sublinear convergence rate of HOPM, we will repeat the proof in~\cite{Hu2018} Theorem~3.2,  which was originally done for the real-valued tensors.
We will use the objective function $F: \mathbb{R}^{2m_1} \times \ldots \times \mathbb{R}^{2m_n} \to \mathbb{R}$ defined for HOPM for an $n$-order complex valued tensors $A \in \mathbb{C}^{m_1} \otimes \ldots \otimes \mathbb{C}^{m_n}$:
\begin{equation}\label{eqn:hopm-obj-func}
    F({\bf x}) \coloneq - \Re(A \times \complMap({\bf x})^*) + \sum_{i=1}^{n} \delta_{S_{2m_i-1}} ({\bf x}_i),
\end{equation}
where ${\bf x} = ({\bf x}_1, \ldots, {\bf x}_n)$, with ${\bf x}_i = (\Re(v_{i,0}), \Re(v_{i,1}), \ldots, \Im(v_{i,0}), \Im(v_{i,1}), \ldots)^{\rm T}$ being real-vector representations (decomplexifications; see~\cite{KostrikinBook}) of the complex vectors ${\bf v}_i = (v_{i,0}, v_{i,1}, \ldots, v_{i,m_i-1})^{\rm T} \in \mathbb{C}^{m_i}$; the vector-function $\complMap:\mathbb{R}^{2m_1} \times \ldots \times \mathbb{R}^{2m_n} \to \mathbb{C}^{m_1} \times \ldots \times \mathbb{C}^{m_n}$ returns the complex representation (complexification) of a given vector or a vector tuple: 
\[
\complMap({\bf x}) = \left(
\left(
\begin{matrix}
    {x}_{1,0} \\
    \vdots \\
    {x}_{1,m_1-1}
\end{matrix}
\right)
+ \imi \left(
\begin{matrix}
    {x}_{1,m_1} \\
    \vdots \\
    {x}_{1,2m_1-1}
\end{matrix}
\right),
\ldots ,
\left(
\begin{matrix}
    {x}_{n,0} \\
    \vdots \\
    {x}_{n,m_n-1}
\end{matrix}
\right)
+ \imi \left(
\begin{matrix}
    {x}_{n,m_n} \\
    \vdots \\
    {x}_{n,2m_n-1}
\end{matrix}
\right)
\right);
\]
$\delta_{S_{2m_i-1}}$ is an indicator function of the unit sphere $S_{2m_i-1}$ in $\mathbb{R}^{2m_i}$, defined as
\[
    \delta_{S_{2m_i-1}}({\bf a}) \coloneq \left\{ 
    \begin{matrix}
        0 & {\rm if} & {\bf a} \in \mathbb{R}^{2m_i}, \| {\bf a} \| = 1, \\
        +\infty & {\rm else}. &
    \end{matrix}
    \right.
\]
It is easy to see that the minimization of the semi-continuous polynomial function $F$ is equivalent to the minimization problem in Eq.~\eqref{eqn:gm_obj_func}.

Let $\lambda^{(k)}_i \coloneq \| {\bf u}_i^{(k)} \|  = | \langle {\bf u}_n^{(k)} , {{\bf v}_n^{(k)}} \rangle |$ ($\lambda ^{(k+1)}_0 \equiv \lambda ^{(k)}_n \equiv \lambda^{(k)}$) be a norm of the $i$th vector update on $k$th iteration in HOPM (see line~6 in Algorithm~\ref{alg:hopm}), where $\langle \cdot,\cdot \rangle$ and $\| \cdot \|$ are inner product and the $l^2$-norm on a complex Hilbert space $\mathbb{C}^{m_i}$. We will abuse these notations and use them further for the inner product and the norm on the real vector spaces also. 
In terms of the ALS method, which HOPM is built upon, for some $1 \leq i \leq n$ and $k>0$, the vector ${\bf v}_i^{(k)}$ is calculated as a solution to the given optimization problem of maximizing $\lambda$, having other vectors fixed (see Algorithm~\ref{alg:hopm}), thus the obtained value $\lambda^{(k)}_i$ should be an improvement over the previously calculated value $\lambda^{(k)}_{i-1}$, so $\lambda^{(k)}_i$ should be not less than $\lambda^{(k)}_{i-1}$. As it is true for any $1 \leq i \leq n$ and $k>0$, we see that the sequence $\{ \lambda^{(1)}_1, \ldots, \lambda^{(1)}_n, \ldots, \lambda^{(k)}_1, \ldots, \lambda^{(k)}_n, \ldots \}$ is monotonically convergent:
\begin{equation}\label{eq:hopm-monotonic}
    0 < \lambda^{(1)}_1 \leq \ldots \leq \lambda^{(1)}_n \leq \ldots \leq \lambda^{(k)}_1 \leq \ldots \leq \lambda^{(k)}_n \leq \ldots \leq |A|,
\end{equation}
where the last inequality comes from Cauchy-Schwartz's inequality for the line~6 in Algorithm~\ref{alg:hopm}.

Let $G({\bf x}) \coloneq A \times \complMap({\bf x})^*$, and 
\begin{equation}\label{eqn:complixified-grad}
    \complMap{(\nabla_i \Re (G({\bf x}^{(k,i)})))} 
    = A \times_i ({\bf v}^{(k)}_1,\ldots,{\bf v}^{(k)}_{i-1}, {\bf v}^{(k-1)}_{i+1}, \ldots, {\bf v}^{(k-1)}_n)^* 
    \equiv A \circ_i \tau_i({\bf v}^{(k,i)})^*,
\end{equation}
where $\nabla_i$ denote a partial gradient with respect to the $i$th vector ${\bf x}_i^{(k)}$ in the tuple ${\bf x}^{(k,i)}$ for any $k$; $\tau_i({\bf x}) = {\bf x}_1 \otimes \ldots \otimes {\bf x}_{i-1} \otimes {\bf x}_{i+1} \otimes \ldots \otimes {\bf x}_{n}$ is a function, which maps the vector tuple ${\bf x}$ into a tensor product space, skipping the $i$th vector; ${\bf x}^{(k,i)} = ({\bf x}^{(k)}_1,\ldots,{\bf x}^{(k)}_{i-1}, {\bf x}^{(k)}_{i}, {\bf x}^{(k-1)}_{i+1}, \ldots, {\bf x}^{(k-1)}_n)$; the $\circ_i$ sign denote a tensor contraction over all indices skipping the $i$th index of the left-hand side tensor.
For $\Re(G({\bf x}))$ we also have:
\begin{equation}\label{eqn:hopm-gx-eq-grdaig}
    \nabla_i \Re (G({\bf x}^{(k,i)}))  = \Re (G({\bf x}^{(k,i)})) {\bf x}^{(k)}_i 
    = \Re (\langle {\bf u}_i^{(k)} , \complMap({\bf x}^{(k)}_i) \rangle ) {\bf x}^{(k)}_i
    = \lambda^{(k)}_i {\bf x}^{(k)}_i.
\end{equation}
The left-hand side and the right-hand side of this equation define a  system of non-linear equations $\nabla_i \Re (G({\bf x}^{(k,i)})) = \lambda^{(k)}_i {\bf x}^{(k)}_i$ w.r.t. real-valued vectors ${\bf x}$ and $\lambda^{(k)}_i = |\Re (\langle {\bf u}_i^{(k)} , \complMap({\bf x}^{(k)}_i) \rangle )|$  instead of the system \meqref{eqn:nonlinear-system} w.r.t. ${\bf v}$ and $\lambda^{(k)}_i = |\langle {\bf u}_i^{(k)} , \complMap({\bf x}^{(k)}_i) \rangle |$. 
These systems of equations are equivalent, since they target the same cost function w.r.t. the same arguments only in different representations. This equivalence can be explicitly obtained by applying the map $\complMap$ to \meqref{eqn:hopm-gx-eq-grdaig} and take into account \meqref{eqn:complixified-grad}.

\begin{lemma}\label{lemma:hopm-lambda-inequality}
    For any $k>0$ the following inequality holds:
    \[
        \lambda^{(k+1)} - \lambda^{(k)} \geq  \frac{\lambda^{(1)}_1}{2} \| {\bf x}^{(k+1)} - {\bf x}^{(k)} \|^2.
    \]
\end{lemma}
\begin{proof}
    We start by presenting the difference $\lambda^{(k+1)} - \lambda^{(k)}$ as the telescopic sum:
    \begin{equation}\label{eq:hopm-lambda-telescopic}
        \lambda^{(k+1)} - \lambda^{(k)} = \sum_{i=1}^n (\lambda^{(k+1)}_i - \lambda^{(k+1)}_{i-1}).
    \end{equation}
    For each element in the sum with all ${\bf x}_i^{(k)} \in S_{2m_i - 1}$ the following holds:
    \begin{equation}\label{eq:hopm-lambda-telescopic1}
    \begin{split}
        \lambda^{(k+1)}_i - \lambda^{(k+1)}_{i-1}
        &= \Re(G({\bf x}^{(k+1,i)})) - \Re(G({\bf x}^{(k+1,i-1)})) \\
        &= \langle \nabla_i \Re(G({\bf x}^{(k+1)}_i)), ({\bf x}^{(k+1)}_i - {\bf x}^{(k)}_{i}) \rangle \\
        &= \Re(G({\bf x}^{(k+1,i)})) \langle {\bf x}^{(k+1)}_i, ({\bf x}^{(k+1)}_i - {\bf x}^{(k)}_{i}) \rangle \\
        &= \frac{\lambda^{(k+1)}_i}{2} \| {\bf x}^{(k+1)}_i - {\bf x}^{(k)}_{i} \|^2, 
    \end{split}
    \end{equation}
    where we have taken into account in the second line  
    \[
        \Re(G({\bf x}^{(k+1,i-1)})) = \Re(\langle {\bf u}_{i-1}^{(k+1)}, \complMap({\bf x}_{i-1}^{(k+1)}) \rangle) 
        = \Re(\langle {\bf u}_{i}^{(k+1)}, \complMap({\bf x}_{i}^{(k)}) \rangle) 
        = \langle \nabla_i \Re( G({\bf x}^{(k+1,i)}) ), {\bf x}^{(k)}_i \rangle,
    \] 
    and in the fourth line 
    \[
        {\bf x}^{(k+1)}_i = \frac{1}{2}({\bf x}^{(k+1)}_i - {\bf x}^{(k)}_i) + \frac{1}{2}({\bf x}^{(k+1)}_i + {\bf x}^{(k)}_i).
    \]

    Substituting the result in \meqref{eq:hopm-lambda-telescopic1} back into \meqref{eq:hopm-lambda-telescopic} and taking into account monotonicity of the ALS method \meqref{eq:hopm-monotonic} we get:
    \[
        \lambda^{(k+1)} - \lambda^{(k)} 
        = \sum_{i=1}^n \frac{\lambda^{(k+1)}_i}{2} \| {\bf x}^{(k+1)}_i - {\bf x}^{(k)}_{i} \|^2
        \geq \sum_{i=1}^n \frac{\lambda^{(1)}_1}{2} \| {\bf x}^{(k+1)}_i - {\bf x}^{(k)}_{i} \|^2
        = \frac{\lambda^{(1)}_1}{2} \| {\bf x}^{(k+1)} - {\bf x}^{(k)} \|^2,
    \]
    which proves the claim of the Lemma.
\end{proof}

The results of Lemma~\ref{lemma:hopm-lambda-inequality} together with Eq.~\eqref{eqn:hopm-obj-func} gives
\begin{equation}\label{eqn:hopm-obj-func-ineq}
    F({\bf x}^{(k)}) - F({\bf x}^{(k+1)}) \geq \frac{\lambda^{(1)}_1}{2} \| {\bf x}^{(k)} - {\bf x}^{(k+1)} \|^2.
\end{equation}
    
\begin{lemma}\label{lemma:hopm-obj-grad-ineq}
    The objective function $F$ defined in \meqref{eqn:hopm-obj-func} for any $k>0$ satisfies:
    \begin{equation}\label{eqn:hopm-objf-ineq}
        F({\bf x}^{(k)}) - F({\bf x}^{(k+1)}) \geq \frac{\lambda_1^{(1)}}{2nL^2} \| \nabla \Re(G({\bf x}^{(k)})) - \lambda^{(k)} {\bf x}^{(k)} \|^2,
    \end{equation}
    where $L=(2\sqrt{n}+1)\|A\| > 0$.
\end{lemma}
\begin{proof}

    For any $k\geq0$ we have:
    \begin{equation}\label{eqn:hopm-grdi-ineq}
    \begin{split}
        \| \nabla_i \Re (G({\bf x}^{(k)}) ) - \lambda^{(k)} {\bf x}_{i}^{(k)} \|
        =& \| \nabla_i \Re (G({\bf x}^{(k)})) - \nabla_i \Re( G({\bf x}^{(k+1,i)}) ) \\
        &+ \lambda^{(k+1)}_i {\bf x}_{i}^{(k+1)} - \lambda^{(k+1)}_i {\bf x}_{i}^{(k)} 
        + \lambda^{(k+1)}_i {\bf x}_{i}^{(k)} - \lambda^{(k)} {\bf x}_{i}^{(k)} \| \\
        \leq& \| \nabla_i \Re( G({\bf x}^{(k)}) ) - \nabla_i \Re( G({\bf x}^{(k+1,i)}) ) \| \\
        &+ \| \lambda^{(k+1)}_i ({\bf x}_{i}^{(k+1)} - {\bf x}_{i}^{(k)}) \| 
        + | \lambda^{(k+1)}_i - \lambda^{(k)}  | \| {\bf x}_{i}^{(k)} \|.
    \end{split}
    \end{equation}
    For the difference of the gradients we have:
    \begin{equation}\label{eqn:hopm-real-grad-ineq}
    \begin{split}
        \| \nabla_i \Re(G({\bf x}^{(k)})) - \nabla_i \Re(G({\bf x}^{(k+1,i)})) \| 
        =& \| \nabla_i \Re( G({\bf x}^{(k)}) -  G({\bf x}^{(k+1,i)})) \| \\
        =& \| A \circ_i (\tau_i(\complMap({\bf x}^{(k)})^*) - \tau_i(\complMap({\bf x}^{(k+1,i)})^*)) \| \\
        \leq& \| A \| \| \tau_i(\complMap({\bf x}^{(k)})^*) - \tau_i(\complMap({\bf x}^{(k+1,i)})^*) \| \\
        \leq& \sqrt{n} \| A \| \| \complMap({\bf x}^{(k)})^* - \complMap({\bf x}^{(k+1,i)})^* \| \\
        =& \sqrt{n} \| A \| \| {\bf x}^{(k)} - {\bf x}^{(k+1,i)} \|,
    \end{split}
    \end{equation}
    where in the last inequality we have used Lemma~2.1 from \cite{Hu2018}:
    \[
        \| \tau_i({\bf a}) - \tau_i({\bf b}) \| 
        \leq \sqrt{n} \| {\bf a} - {\bf b} \|,
    \]
    which holds for any ${\bf a}, {\bf b} \in S_{d_1-1}\times \ldots \times S_{d_n - 1}$.

    Substituting the result of \meqref{eqn:hopm-real-grad-ineq} into \meqref{eqn:hopm-grdi-ineq} and taking into account, that $\lambda_i^{(k+1)} = | A \circ_i \tau_i(\complMap({\bf x}^{(k+1,i)})) |$, we have:
    \begin{equation}\label{eqn:hopm-grdi-ineq-final}
    \begin{split}
        \| \nabla_i \Re( G({\bf x}^{(k)}) ) - \lambda^{(k)} {\bf x}_{i}^{(k)} \|
        \leq& 2 \sqrt{n} \| A \| \| {\bf x}^{(k+1,i)} - {\bf x}^{(k)} \|
        + \| A \| \| {\bf x}^{(k+1,i)} - {\bf x}^{(k)} \| \\
        \leq& L \| {\bf x}^{(k+1)} - {\bf x}^{(k)} \|,
    \end{split}
    \end{equation}
    where $L=(2\sqrt{n}+1)\|A\| > 0$. This gives us
    \[
        \| \nabla \Re( G({\bf x}^{(k)}) ) - \lambda^{(k)} {\bf x}^{(k)} \|^2
        = \sum_{i=1}^n \| \nabla_i \Re( G({\bf x}^{(k)}) ) - \lambda^{(k)} {\bf x}_{i}^{(k)} \|^2
        \leq nL^2 \| {\bf x}^{(k)} - {\bf x}^{(k+1)} \|^2.
    \]
    Using this result together with \meqref{eqn:hopm-obj-func-ineq} gives
    \[
        F({\bf x}^{(k)}) - F({\bf x}^{(k+1)}) \geq \frac{\lambda_1^{(1)}}{2nL^2} \| \nabla \Re( G({\bf x}^{(k)}) ) - \lambda^{(k)} {\bf x}^{(k)} \|^2,
    \]
    which finishes the proof.
\end{proof}

\begin{theorem}[Sublinear convergence rate of HOPM for complex valued tensors.]
\label{thrm:hopm-sublin-conv}
    Let $\complMap({\bf x}^{(k)})$ and $\lambda^{(k)} = |A \times \complMap({\bf x}^{(k)})^* |$ be generated by Algorithm~\ref{alg:hopm} for a given tensor $A \in \mathbb{C}^{m_1} \otimes \ldots \otimes \mathbb{C}^{m_n}$ at the $k$th iteration, and let $p = n(3n-3)^{2M}$, where $M = \sum_{i=1}^n m_i$. Then there exists $B>0$ s.t. the following holds:
    \begin{equation}\label{eqn:hopm-compl-conv-rate}
        \hat{\lambda} - \lambda^{(k)} \leq B{\left( \frac{p-2}{n^2 p}k \right)^{-\frac{p}{p-2}} },
    \end{equation}
    where $\hat{\lambda}$ is the entanglement eigenvalue for the tensor $A$.
\end{theorem}
\begin{proof}
    Recall that if $\complMap(\hat{{\bf x}})$ is an eigenvector tuple of the tensor $A$ in the sense of RTA, $\hat{\lambda} \hat{\bf x}_i = \Re( G(\hat{\bf x})) \hat{\bf x}_i = \nabla_i \Re( G(\hat{\bf x}) )$ holds for any $i=1,\ldots,n$. Define a function $H: (\mathbb{R}^{2m_1} \times \ldots \times \mathbb{R}^{2m_n})\times \mathbb{R} \to \mathbb{R}$:
    \[
        H({\bf x}, \mu) = -\Re(G({\bf x})) + \mu \sum_{i=1}^n (\| {\bf x}_i \|^2 - 1)
        = -\Re(A \times \complMap({\bf x})^* ) + \mu \sum_{i=1}^n (\| {\bf x}_i \|^2 - 1).
    \]
    Let $\hat{\mu} = |G(\hat{\bf x})|/2$ and let $\hat{H}({\bf x}, \mu) \coloneq H({\bf x}, \mu) - H(\hat{\bf x}, \hat{\mu})$. It is clear that $\hat{H}$ is a real polynomial on $\mathbb{R}^{2M + 1}$ of degree $n$ with $\hat{H}(\hat{\bf x}, \hat{\mu}) = 0$ and $\nabla \hat{H}(\hat{\bf x}, \hat{\mu}) = 0$:
    \begin{equation*}
        \nabla\hat{H}(\hat{\bf x}, \hat{\mu}) 
        = \nabla H(\hat{\bf x}, \hat{\mu}) 
        = (2\hat{\mu}\hat{\bf x}_1 - \nabla_1 \Re( G(\hat{\bf x})), \ldots,  2\hat{\mu}\hat{\bf x}_1 - \nabla_n \Re(G(\hat{\bf x})))
        = {\bf 0}.
    \end{equation*}
    Applying the {\L}ojasiewicz inequality (Theorem~4.2 in~\cite{DAcunto2005}) it follows that there exist $c, \epsilon >0$ such that 
    \[
        \| \nabla \hat{H}({\bf x}, \mu) \| \geq c | \hat{H}({\bf x}, \mu) |^r, 
        \quad 
        r = \frac{p-1}{p}
    \]
    for all $\| ({\bf x}, \mu) - (\hat{\bf x}, \hat{\mu}) \| \leq \epsilon$. 
    Then for all ${\bf x} \in S_{2m_1-1}\times\ldots\times S_{2m_n-1}$ and $\mu \in \mathbb{R}$ with $\| ({\bf x}, \mu) - (\hat{\bf x}, \hat{\mu}) \| \leq \epsilon$, we get
    \[
        \| -\nabla \Re(G({\bf x})) + 2\mu {\bf x} \|^2 \geq c^2 (H(\hat{\bf x}, \hat{\mu}) - H({\bf x}, \mu))^{2r}.
    \]
    Let $\mu = \Re(G({\bf x}))/2$, then for all ${\bf x} \in S_{2m_1-1}\times\ldots\times S_{2m_n-1}$ with $\| {\bf x} - \hat{\bf x} \| \leq \epsilon$ for some $\epsilon>0$, we have
    \begin{equation}\label{eqn:hopm-grad-F-ineq}
        \| -\nabla\Re(G({\bf x})) + \Re(G({\bf x})) {\bf x} \|^2 \geq c^2 (H(\hat{\bf x}, \hat{\mu}) - H({\bf x}, \mu))^{2r} = c^2 (F({\bf x}) - F(\hat{\bf x}))^{2r}.
    \end{equation}
    As Algorithm~\ref{alg:hopm} converges ${\bf x}^{(k)} \to \hat{\bf x}$ (see~\cite{Zhang2020a}), there exists $k_0$, such that for all $k>k_0$, $\| {\bf x}^{(k)} - \hat{\bf x} \|<\epsilon$. This allows us to use \meqref{eqn:hopm-grad-F-ineq} together with Lemma~\ref{lemma:hopm-obj-grad-ineq} and obtain
    \begin{equation}\label{eqn:hopm-conv-obj-FF-ineq}
        F({\bf x}^{(k)}) - F({\bf x}^{(k+1)}) \geq \frac{\lambda^{(1)}_1 c^2}{2nL^2} (F({\bf x}^{(k)}) - F(\hat{\bf x}))^{2r}.
    \end{equation}

    Denote $\beta_k = (F({\bf x}^{(k)}) - F(\hat{\bf x})) \geq 0$ and $C = \frac{\lambda^{(1)}_1 c^2}{2nL^2}$. Then, 
    \[
        \beta_k - \beta_{k+1} \geq C \beta_k^{2r} = C h(\beta_k)^{-1},
    \]
    where $h(a) = a^{-2r}$. Note that $h(x) = f'(x)$ with $f(x) = \frac{x^{1-2r}}{1-2r}$ and $h(x)$ is non-increasing on the positive real numbers $\mathbb{R}_+$. This allows us to get the following chain of inequalities:
    \[
        C \leq h(\beta_k)(\beta_k - \beta_{k+1}) \leq \int\limits_{\beta_{k+1}}^{\beta_{k}} h(x)dx = f(\beta_k) - f(\beta_{k+1}) = \frac{1}{2r-1} (\beta_{k+1}^{1-2r} - \beta_{k}^{1-2r}),
    \]
    which gives
    \[
        \beta_{k}^{1-2r} \geq C(2r-1) + \beta_{k-1}^{1-2r} 
        \geq \ldots \geq 
        C(2r-1)(k-k_0) + \beta_{k_0}^{1-2r}.
    \]
    Hence, taking into account the definition of $C>0$, there exists $B>0$ such that for all $k \geq k_0$ we have 
    \[
        0 \leq \beta_{k} \leq B\left(\left( \frac{p-2}{n^2 p}k \right)^{-\frac{p}{p-2}} \right).
    \]
    Note that $\beta_{k} = F({\bf x}^{(k)}) - F(\hat{\bf x}) = \hat{\lambda} - \lambda^{(k)} \geq 0$, then for all $k > k_0$ we have:
    \[
        \hat{\lambda} - \lambda^{(k)} \leq B{\left( \frac{p-2}{n^2 p}k \right)^{-\frac{p}{p-2}} }.
    \]
\end{proof}

\section{Connection of the variance of the cost function gradient in VDGE/iVDGE methods with QHOPM}
\label{app:ivdge-scaling}

A cost function $I(\boldsymbol{\theta}) = I(\theta_1, \ldots, \theta_n)$ in terms of a standard gradient descent method after update of a subset of its parameters $\boldsymbol{\theta}' = (\theta_1,\ldots, \theta_m)$ ($m\leq n$) changes as follows:
\begin{equation}\label{eqn:cost-func-change}
    I(\boldsymbol{\theta}' - \eta \nabla_{\boldsymbol{\theta}'} I, \theta_{m+1}, \ldots, \theta_n) - I(\boldsymbol{\theta}) 
    \approx -\eta |\nabla_{\boldsymbol{\theta}'} I|^2 + \ldots
    = -\eta \sum_{i=1}^m \left(\frac{I(\boldsymbol{\theta})}{\partial \theta_i} \right)^2 + \ldots,
\end{equation}
where $\eta \in \mathbb{R}_+$ is usually referred to as learning rate.
Recall in complexity analysis we consider the case where qubit number $n \to \infty$ and iteration $k \to \infty$. In these limits, we can neglect higher-order terms in the decomposition in \meqref{eqn:cost-func-change}. After averaging over the parameters $\boldsymbol{\theta}'$ we get 
\begin{equation}\label{eqn:cost-func-change-avg}
    \mathbb{E}_{\boldsymbol{\theta}'}[I(\boldsymbol{\theta}' - \eta \nabla_{\boldsymbol{\theta}'} I, \theta_{m+1}, \ldots, \theta_n) - I(\boldsymbol{\theta})] 
    \approx -\eta m \mathbb{E}_{{\theta}_i}\left[\left(\frac{I(\boldsymbol{\theta})}{\partial \theta_i} \right)^2\right] .
\end{equation}
In the paper on iVDGE~\cite{Zambrano2024} it was shown that the average of the gradient over any angle is zero in the case of VDGE and iVDGE ($\mathbb{E}_{\theta_i}[\partial I(\boldsymbol{\theta})/\partial\theta_i]=0$ (see \meqref{eq:hopm-lambda-telescopic} and \meqref{eqn:hopm-real-grad-ineq} in~\cite{Zambrano2024}). 
So for VDGE and iVDGE the difference in \meqref{eqn:cost-func-change-avg} is the variance of the gradient of the cost function ($\mathrm{Var[\partial I(\boldsymbol{\theta})/\partial\theta_i]}$) with $(-\eta m)$ prefactor.

In case of QHOPM, the cost function is defined in \meqref{eqn:hopm-obj-func}
and for normalized vectors on the $k$-th iteration it equals to the negative eigenvalue $\lambda^{(k)}$: $F({\bf x}^{(k)}) = -\lambda^{(k)}$. 
Thus, a difference of eigenvalues on the $k$-th and $(k-1)$-th iterations is equivalent to the difference of cost functions, when two parameters (two angles for $n$-th qubit) were updated.  For VDGE and iVDGE in the limit of large $n$ and $k$ this corresponds to the variance of the cost function gradient with $m=2$ as shown above.

\section{Implementation of Gauss-Seidel method and Shifted HOPM algorithms}
\label{app:gauss-siedel}
The approach presented in this paper allows us to also implement the
Gauss-Seidel method (GSM), which (as authors of the work~\cite{Zhang2020a}
claim) should converge to the global minimum with higher probability.
The difference with HOPM is in line 6 of Algorithm~\ref{alg:hopm}:
it is changed to
\begin{equation}\label{eqn:zhang-line6}
  {\mathbf{u}}_i^{(k+1)}= \lambda^{(k)} T_\psi \times_{{i}} \left(
    {\mathbf{v}_1^{(k+1)}}^\ast, \ldots
    ,{\mathbf{v}_{i-1}^{(k+1)}}^\ast, {\mathbf{v}_{i+1}^{(k)}}^\ast,
  \ldots, {\mathbf{v}_n^{(k)}}^\ast \right) + \beta {\mathbf{v}}_i^{(k)},
\end{equation}
where $0 < \beta \in \mathbb{R} $ is a free parameter of the algorithm.

To implement GSM the only change to our algorithm is in the
classical step where we
recover the angles
$(\check{\vartheta}_{i}^{(k+1)}, \check{\varphi}_{i}^{(k+1)})$
using one-qubit tomography.
Instead of \meqref{eq:tomography-system}
we need to solve the following system of equations to find
$({\vartheta}_i^{(k+1)}, {\varphi}_i^{(k+1)})$,
\begin{equation}\label{eq:1q-zhang33-tomogr}
  \begin{aligned}
    A^2  \cos& ^2 \frac{\vartheta_i^{(k+1)}}{2} =
    \left| \lambda^{(k)} e^{-i \check{\varphi}_{i}^{(k+1)}/2}
    \cos\frac{\check{\vartheta}_{i}^{(k+1)}}{2} \right.
    + \left. \beta e^{-i \varphi_i^{(k)}/2} \cos
    \frac{\vartheta_i^{(k)}}{2} \right|^2;
    \\
    A^2  \sin & \vartheta_i^{(k+1)}  \sin \varphi_i^{(k+1)} =
    (\lambda^{(k)})^2 \sin \check{\vartheta}_{i}^{(k+1)} \sin
    \check{\varphi}_{i}^{(k+1)} \\
    &+ \beta^2 \sin \vartheta_i^{(k)} \sin \varphi_i^{(k)}
    + 2\beta \lambda^{(k)} \sin \frac{\check{\vartheta}_{i}^{(k+1)}
    + \vartheta_i^{(k)}}{2} \sin\frac{\check{\varphi}_{i}^{(k+1)} +
    \varphi_i^{(k)}}{2};                   \\
    A^2  \sin & \vartheta_i^{(k+1)}  \cos \varphi_i^{(k+1)}    =
    (\lambda^{(k)})^2 \sin \check{\vartheta}_{i}^{(k+1)} \cos
    \check{\varphi}_{i}^{(k+1)} \\
    &+ \beta^2 \sin \vartheta_i^{(k)} \cos \varphi_i^{(k)}
    + 2\beta \lambda^{(k)} \sin \frac{\check{\vartheta}_{i}^{(k+1)}
    + \vartheta_i^{(k)}}{2} \cos\frac{\check{\varphi}_{i}^{(k+1)} +
    \varphi_i^{(k)}}{2},                   \\
  \end{aligned}
\end{equation}
where
\[
A^2 = (\lambda^{(k)})^2+\beta^2 + 2\beta \lambda^{(k)} \cos
\frac{\check{\vartheta}_{i}^{(k+1)} - \vartheta_i^{(k)}}{2} \cos
\frac{\check{\varphi}_{i}^{(k+1)} - \varphi^{(k)}}{2}
\]
is a normalization coefficient.

Note, that GSM is similar to the so-called shifted HOPM (SHOPM)
algorithm~\cite{Kolda2011, Hu2016}, with the only difference being
that SHOPM does not have the current eigenvalue estimate
$\lambda^{(k)}$ as a prefactor on the right-hand side of
\meqref{eqn:zhang-line6}. Thus to get the SHOPM quantum-classical
hybrid implementation one needs to set $\lambda^{(k)}$ in
\meqref{eq:1q-zhang33-tomogr} to one.

\section{Supplemental simulation results}
\label{app:supplementary-sims}

\begin{figure*}
  \makebox[\textwidth][c]{
    \includegraphics{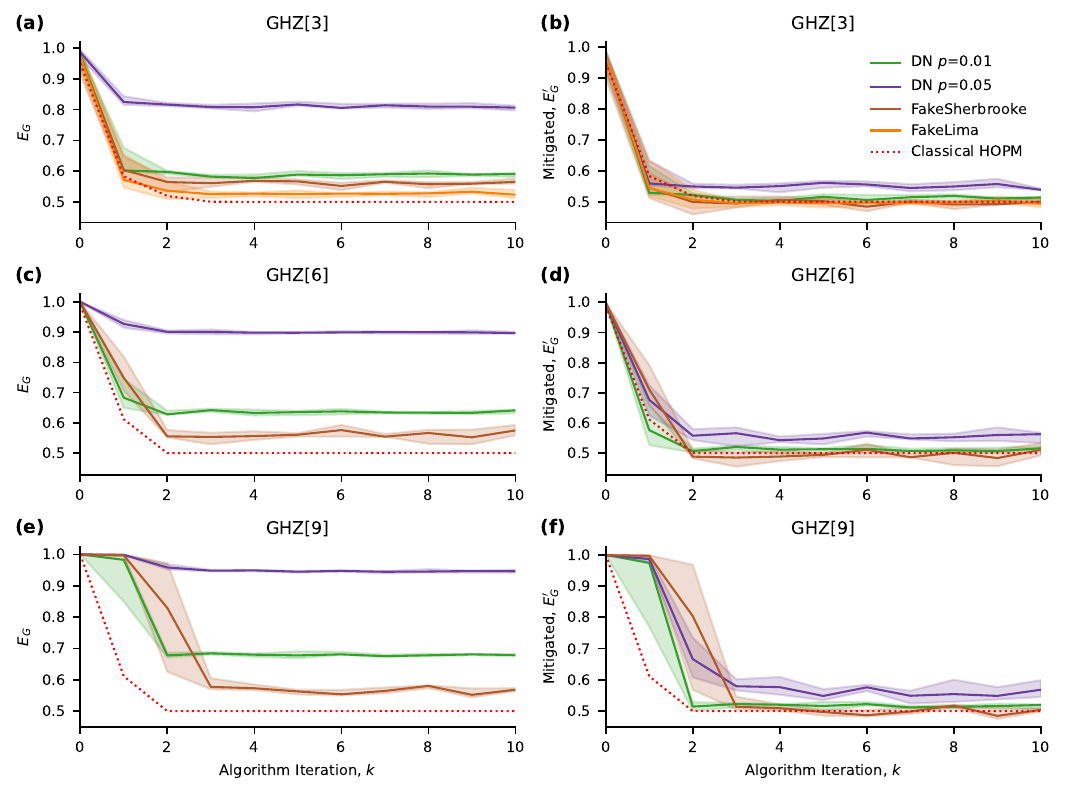}
  }
  \caption{\textbf{Simulation and mitigation results for the QHOPM
    algorithm on GHZ states.}
    Row \textbf{(a)} \textbf{(b)}  show \GHZ{3},
    \textbf{(c)} \textbf{(d)}  show \GHZ{6},
    \textbf{(e)} \textbf{(f)}  show \GHZ{9}.
    Column \textbf{(a)}  \textbf{(c)} \textbf{(e)}  shows the
    effect of noise on convergence.
    Column \textbf{(b)}  \textbf{(d)} \textbf{(f)}  shows the
    effects of mitigation on the noisy simulation.
    Each simulation was run 10 times (with the same 10 random initial
    separable states) with $1\times 10^5$ shots per each measurement.
    Solid lines represent the mean geometric entanglement, $\GE{}$,
    and the lightly
    shaded colours represent the standard deviation.
    The colour of the lines represents the noise model used for the
    simulation (see legend).
    Red dotted lines show results of classical HOPM\@.
  }
  \label{supfig:all_qubits_with_DN_mitigation_GHZ}
\end{figure*}

\begin{figure*}
  \makebox[\textwidth][c]{
    \includegraphics{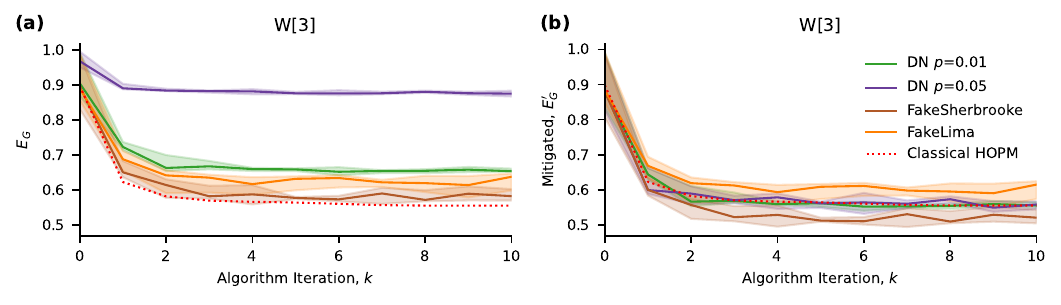}
  }
  \caption{\textbf{Simulation and mitigation results for the QHOPM
    algorithm on the \W{3} state.}
    \textbf{(a)} shows the effect of noise on convergence.
    \textbf{(b)} shows the effects of mitigation on the noisy simulation.
    $\trueGE = 5/9 \approx 0.554$.
    Each simulation was run 10 times (with the same 10 random initial
    separable states) with $1\times 10^5$ shots per each measurement.
    Solid lines represent the mean geometric entanglement, $\GE{}$,
    and the lightly
    shaded colours represent the standard deviation.
    The colour of the lines represents the noise model used for the
    simulation (see legend).
    Red dotted lines show results of classical HOPM\@.
  }
  \label{supfig:all_qubits_with_DN_mitigation_W}
\end{figure*}

\begin{figure*}
  \makebox[\textwidth][c]{
    \includegraphics{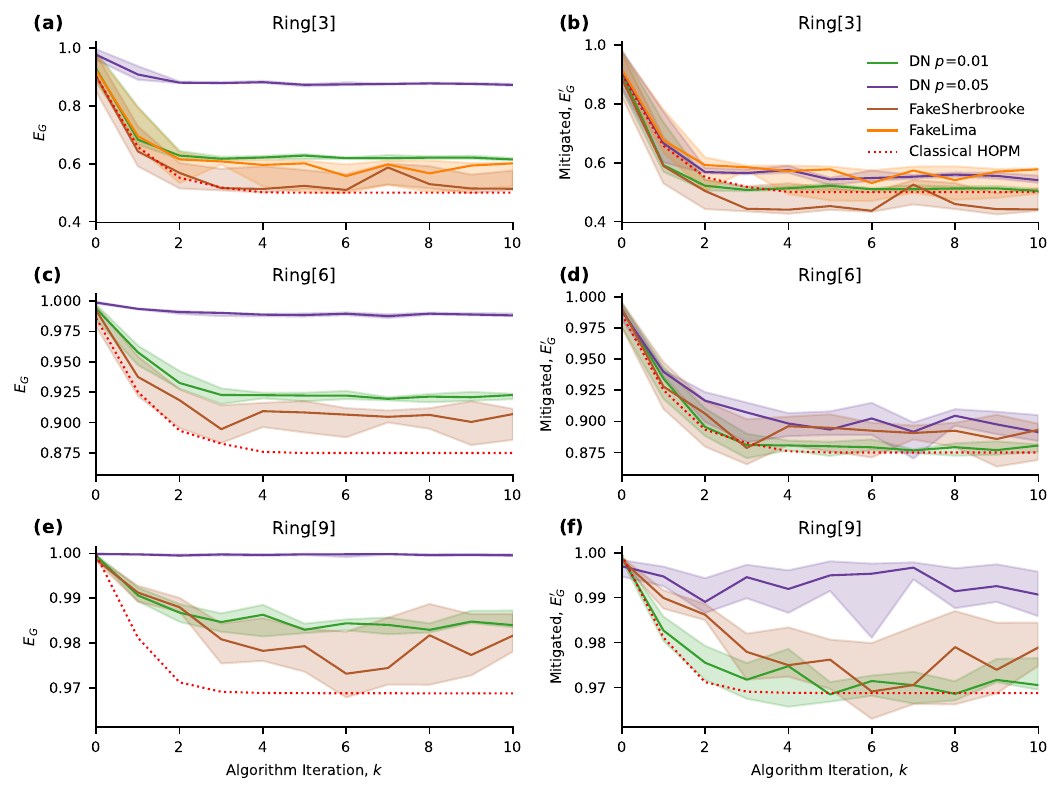}
  }
  \caption{\textbf{Simulation and mitigation results for the QHOPM
    algorithm on Ring states.}
    Row \textbf{(a)} \textbf{(b)}  show \Ring{3},
    \textbf{(c)} \textbf{(d)}  show \Ring{6},
    \textbf{(e)} \textbf{(f)}  show \Ring{9}.
    Column \textbf{(a)}  \textbf{(c)} \textbf{(e)}  shows the
    effect of noise on convergence.
    Column \textbf{(b)}  \textbf{(d)} \textbf{(f)}  shows the
    effects of mitigation on the noisy simulation.
    Each simulation was run 10 times (with the same 10 random initial
    separable states) with $1\times 10^5$ shots per each measurement.
    Solid lines represent the mean geometric entanglement, $\GE{}$,
    and the lightly
    shaded colours represent the standard deviation.
    The colour of the lines represents the noise model used for the
    simulation (see legend).
    Red dotted lines show results of classical HOPM\@.
  }
  \label{supfig:all_qubits_with_DN_mitigation_Ring}
\end{figure*}

\begin{figure*}
  \makebox[\textwidth][c]{
    \includegraphics{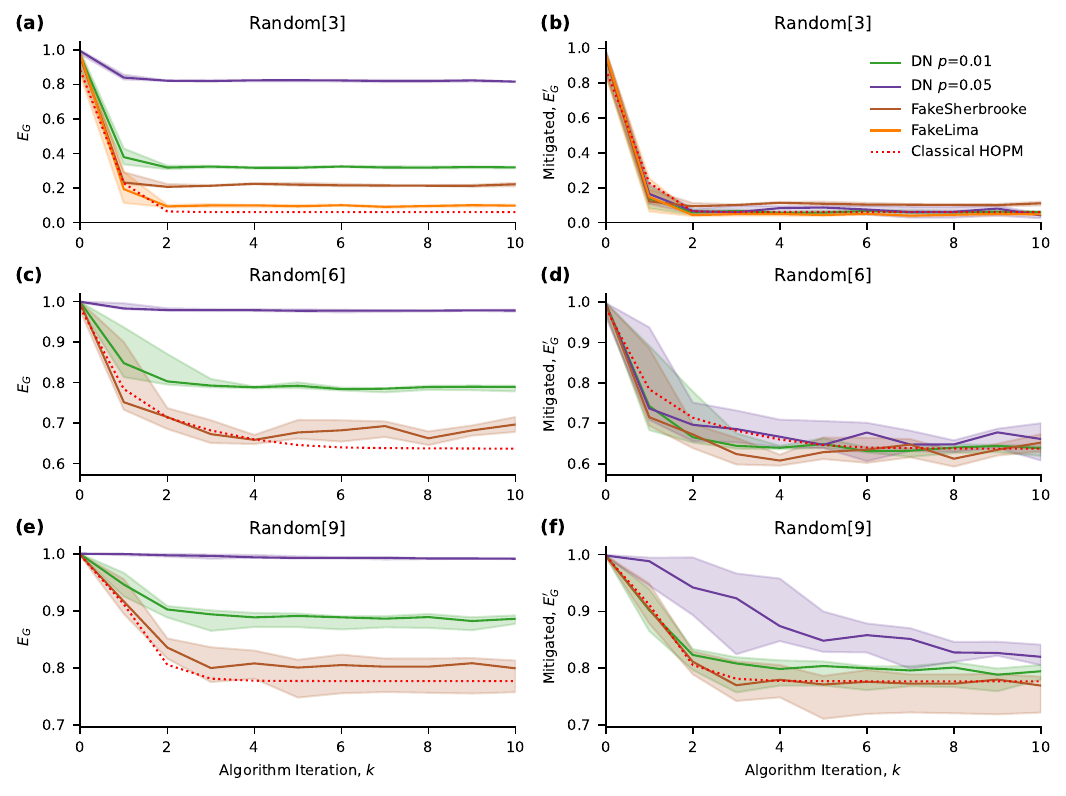}
  }
  \caption{\textbf{Simulation and mitigation results for the QHOPM
    algorithm on Random states.}
    Row \textbf{(a)} \textbf{(b)}  show \Random{3},
    \textbf{(c)} \textbf{(d)}  show \Random{6},
    \textbf{(e)} \textbf{(f)}  show \Random{9}.
    Column \textbf{(a)}  \textbf{(c)} \textbf{(e)}  shows the
    effect of noise on convergence.
    Column \textbf{(b)}  \textbf{(d)} \textbf{(f)}  shows the
    effects of mitigation on the noisy simulation.
    Each simulation was run 10 times (with the same 10 random initial
    separable states) with $1\times 10^5$ shots per each measurement.
    Solid lines represent the mean geometric entanglement, $\GE{}$,
    and the lightly
    shaded colours represent the standard deviation.
    The colour of the lines represents the noise model used for the
    simulation (see legend).
    Red dotted lines show results of classical HOPM\@.
  }
  \label{supfig:all_qubits_with_DN_mitigation_random}
\end{figure*}

\newpage

\end{document}